\documentclass[%
 reprint,
 amsmath,amssymb,
 aps,
]{revtex4-2}

\usepackage{graphicx}
\usepackage{dcolumn}
\usepackage{bm}
\usepackage{hyperref}
\usepackage{amsthm}
\newtheorem{proposition}{Proposition}

\usepackage{settings_custom}
\newcommand{\LAS}{{\mathcal{LAS}}} 
\newcommand{\IGP}{{\mathcal{IGP}}} 
\newcommand{\ain}{{\alpha_{\rm in}}}
\newcommand{\aout}{{\alpha_{\rm out}}}
\newcommand{\kr}{{k_{\rm r}}}
\newcommand{\kc}{{k_{\rm c}}}
\newcommand{\Nr}{{N_{\rm r}}}
\newcommand{\Nc}{{N_{\rm c}}}
\newcommand{\mr}{{m_{\rm r}}}
\newcommand{\mc}{{m_{\rm c}}}
\newcommand{\hr}{{h_{\rm r}}}
\newcommand{\hc}{{h_{\rm c}}}
\newcommand{\qr}{{q_{\rm r}}}
\newcommand{\qc}{{q_{\rm c}}}
\newcommand{\hmr}{{\hat{m}_{\rm r}}}
\newcommand{\hmc}{{\hat{m}_{\rm c}}}
\newcommand{\hqr}{{\hat{q}_{\rm r}}}
\newcommand{\hqc}{{\hat{q}_{\rm c}}}

\usepackage{tikz}

\begin{document}

\title{Statistical mechanics of the maximum-average submatrix problem}

\author{Vittorio Erba$^1$}
\author{Florent Krzakala$^2$}%
\author{Rodrigo Pérez$^2$}%
\author{Lenka Zdeborov\'a$^1$}
\affiliation{%
$^1$Statistical Physics of Computation Laboratory, \\
$^2$Information, Learning and Physics Laboratory, \\
École polytechnique fédérale de Lausanne (EPFL)
CH-1015 Lausanne
}%

\date{\today}

\begin{abstract}
We study the maximum-average submatrix problem, in which given an $N \times N$ matrix $J$ one needs to find the $k \times k$ submatrix with the largest average of entries.
We study the problem for random matrices $J$ whose entries are i.i.d. random variables 
by mapping it to a variant of the Sherrington-Kirkpatrick spin-glass model at fixed magnetization.
We characterize analytically the phase diagram of the model as a function of the submatrix average and the size of the submatrix $k$ in the limit $N\to\infty$.
We consider submatrices of size $k = m N$ with $0 < m < 1$.
We find a rich phase diagram, including dynamical, static one-step replica symmetry breaking and full-step replica symmetry breaking. 
In the limit of $m \to 0$, we find a simpler phase diagram featuring a frozen 1-RSB phase, where the Gibbs measure is composed of exponentially many pure states each with zero entropy.
We discover an interesting phenomenon, reminiscent of the phenomenology of the binary perceptron: there exist efficient algorithms that provably work in the frozen 1-RSB phase.
\end{abstract}

\maketitle

We consider the maximum-average submatrix (MAS) problem, i.e. the problem of finding the $k \times k$ submatrix of an $N \times N$ matrix $J$ with the largest average of entries.
This is a natural combinatorial optimization problem that has been studied in the mathematical and data science literature \cite{LAS}, mainly in the context of biclustering~\cite{biclustering_survey}. Theoretical works focused on the case where $J$ is a random matrix ($i.i.d.$ standard Gaussian entries), the size of $J$ is large ($N \to \infty$) and the size of the submatrix $k \ll N$ \cite{Bhamidi2017,Gamarnik2018,sun2013maximal}.

From a statistical physics point of view, the MAS problem is a natural variant of the well-known Sherrington-Kirkpatrick model \cite{sherrington1975solvable} with spins $\sigma_i \in \{0,1\}$ 
at fixed magnetization. Statistical physics of disordered systems and the related replica method \cite{mezard1987spin} have been used widely to study other combinatorial optimization problems such as 
graph partitioning \cite{fu1986application}, matching \cite{mezard1986mean}, graph colouring \cite{zdeborova2007phase}, $K$-satisfiability of Boolean formulas \cite{mezard2002analytic}, and many others. As far as we are aware, the maximum-average submatrix problem has not been studied from the statistical physics point of view. Filling this gap is the main purpose of the present paper.

\begin{figure*}
  \centering
  \includegraphics[width=2\columnwidth]{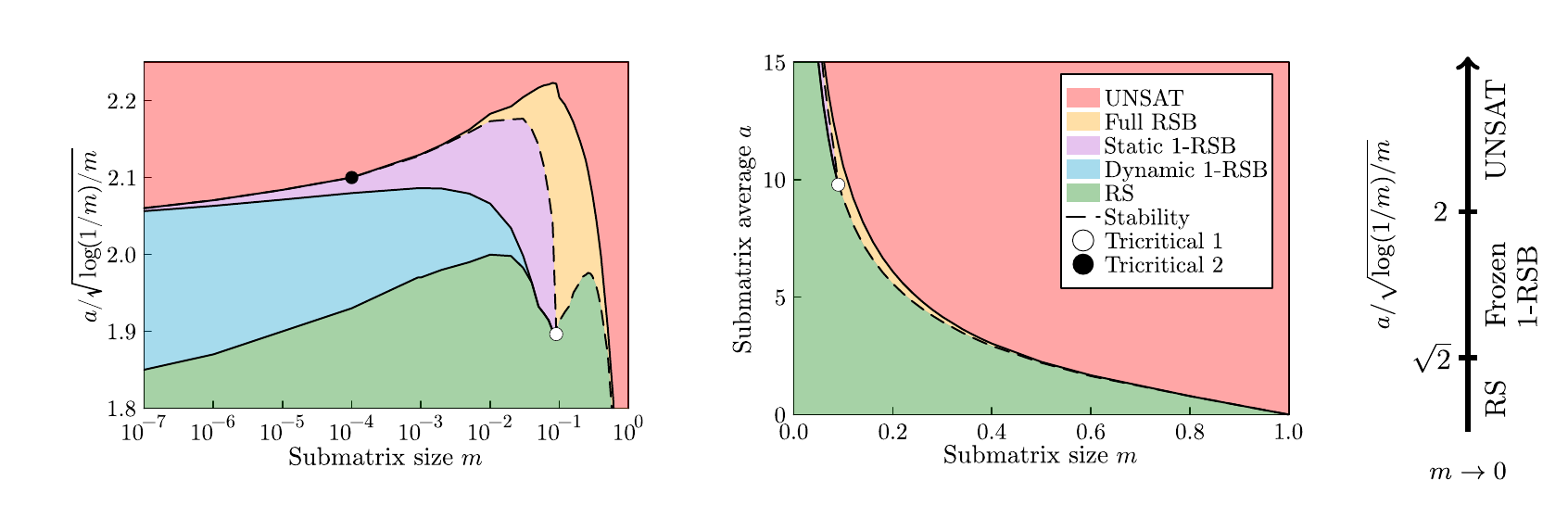}

  \caption{
    The phase diagram of the MAS problem as a function of the submatrix-average $a$ and the submatrix size $m = k/N$ for linear scale in $m$ (left), logarithmic scale in $m$ (center) and $m \to 0$ (right).
    In the central and right panel we rescale the sub-matrix average as $a / \sqrt{\log(1/m)/(m)}$ to highlight the convergence to the limit. 
    We identify five distinct phases. 
    In the RS phase (green) the system is replica symmetric.
    In the 1-RSB phases replica symmetry is broken to one step and two sub-phases exists, a dynamical 1-RSB with an extensive number of equilibrium pure states (blue) and a static 1-RSB with only finitely many pure state (purple). All the phase boundaries 
    are exact in the thermodynamic limit except the boundary between full-RSB (orange) and UNSAT (red) which would require solving the full-RSB equations. 
    In the full-RSB phase (orange) replica symmetry is completely broken and the set of pure states manifests ultrametricity.
    In the unsatisfiable phase (UNSAT, red) no submatrix exists with the given values of $a$ and $m$.
    The transition from RS and 1-RSB to full-RSB is continuous and caused by an instability of the 1-RSB ansatz (dashed line), while the other transitions are discontinuous. 
    We observe two tricritical points, one at $(m_c, a_c)$ where the system shows coexistence of RS, 1-RSB and full-RSB phase (white marker), and one at $(m^*, a^*)$ where the largest-average submatrices become 1-RSB stable and the full-RSB region ceases to exist (black marker).
    In the limit $m\to 0$, we observe only the RS, 1-RSB and UNSAT phases. The 1-RSB phase is frozen, meaning that the internal entropy of each pure state goes to zero in the $m\to 0$ limit.
  }
  \label{fig.ph}
\end{figure*}

Our results reveal the exact phase diagram of the MAS problem when $k=mN$, $N$ is large and $m$ finite. We unveil that at large values of $m$ as the submatrix average increases the system undergoes a continuous phase transition to a full replica symmetry breaking (RSB) phase \cite{mezard1987spin}. At intermediate values of $m$ the phase transition becomes discontinuous, passing through a dynamical one-step RSB and static one-step RSB phases to a full-RSB one \cite{mezard1987spin}. At yet lower $m$, the full-RSB phase then vanishes and the maximum average is given by the one-step RSB solution. In the limit of $m\to 0$, the MAS problem behaves in a way related to the random energy 
model
\cite{derrida1981random} presenting frozen one-step RSB \cite{krauth1989storage,martin2004frozen}.

We also find that in the limit $m\to 0$ the phase diagram presents a region where polynomial algorithms are proven to work \cite{Gamarnik2018} yet according to our results the equilibrium behaviour of the problem is given by the frozen one-step RSB phase that is considered algorithmically hard \cite{gamarnik2022disordered}. One other such problem is known in the literature -- the binary perceptron. For the binary perceptron, an explanation of the discrepancy between equilibrium properties and algorithmic feasibility has been
proposed in relation to out-of-equilibrium large-local-entropy regions of the phase space that are not described within the standard replica solution \cite{baldassi2015subdominant}. This finding has been used to discuss learning in artificial neural networks \cite{baldassi2016unreasonable} and to propose new algorithms \cite{chaudhari2019entropy}. The analogy of behaviour between the binary perceptron and the MAS problem is therefore interesting as it may serve to shed more light on the fundamental question of 
algorithmic hardness. 
From the point of view of the mathematics of spin glasses, the perceptron problem is difficult to handle due to the effective bipartite structure of the correlations. The MAS problem
belongs instead to a class of problems for which the exactness of the replica calculation has been established rigorously in \cite{Panchenko}.

We now review the mathematical results which we later connect to our analysis. 
All these results hold in the regime $k \ll N$. 
In \cite{Bhamidi2017}, the authors proved that the globally optimal submatrix has an average equal to $A_{\rm opt} = 2\sqrt{\log N / k}$ ($\log$ is the natural logarithm). 
They also conjectured, and it was later proven by \cite{Gamarnik2018}, that the Largest Average Submatrix ($\LAS$) algorithm --- an efficient iterative row/column optimization scheme 
--- fails to reach the global optimum, as its fixed point 
--- akin to local minima ---
has with high probability average equal to $A_{\LAS} = \sqrt{2 \log N / k}$. 
The fact that the $\LAS$ algorithm fails bears a natural question: 
is $A_{\LAS}$ an algorithmic threshold signalling the onset of a hard phase?
In \cite{Gamarnik2018}, the authors introduced a new algorithm called Incremental Greedy Procedure ($\IGP$) which is able to produce submatrices with average $A_{\IGP} = 4/3 \sqrt{2 \log N / k} > A_{\LAS}$.
Additionally, they proved that for averages larger than at least $A_{OGP} = 10/(3\sqrt{3}) \sqrt{\log N / k} > A_{\IGP}$ the problem satisfies the Overlap Gap Property (OGP) \cite{reviewGamarnik}. 
This means that i) the  $\LAS$ threshold $A_{\LAS}$ seems to be an algorithm-specific threshold, and not a more general trace of an intrinsic computational-to-statistical gap, and ii) that the problem will likely exhibit a hard phase preventing algorithms to find submatrices with averages larger than at least $A_{\rm OGP}$.
\cite{Cheairi} discusses additional results for $k = N^{\gamma}$, $0 < \gamma < 1$.

\section{The model}
We consider an $\Nr \times \Nc$ random matrix $J$ composed of i.i.d. Gaussian entries with zero mean and unit variance.
A $\kr \times \kc$ submatrix $\sigma$ of $J$ is defined by two arbitrary subsets of rows and columns $I_r, I_c$ such that $J_{ij}$ belongs to the submatrix $\sigma$ if and only if $i \in I_r$ and $j \in I_c$, and the cardinalities satisfy $|I_r| = k_r$ and $|I_c| = \kc$. 
There are three versions of the MAS problem:
\begin{itemize}
    \item \textbf{Rectangular MAS:} $\Nr$, $\Nc$, $\kr$ and $\kc$ are unconstrained. This is the problem relevant for applications \cite{LAS, biclustering_survey}.
    \item \textbf{Square MAS of a square matrix:} $\Nr = \Nc = N$ and $\kr = \kc = k$. This version is studied in the mathematical literature \cite{Bhamidi2017, Gamarnik2018, Cheairi}.
    \item \textbf{Principal MAS of a symmetric matrix:}   $\Nr = \Nc = N$, $\kr = \kc = k$, $J = J^T$ is symmetric and we consider only \textit{principal submatrices}, i.e. submatrices for which $I_r = I_c$.
\end{itemize}

In the following, we focus on the case of the principal MAS of a symmetric random matrix motivated by its close relation to the SK model. 
In the SM, we sketch the corresponding solution for the rectangular MAS of a random matrix and realize with surprise that the equations leading to the phase diagram of the square MAS of a square random matrix are exactly the same as the ones for the principal MAS of a symmetric random matrix. This allows us to compare our results directly to the mathematical literature, and it also means that the phase diagram we provide applies to the non-symmetric case. 
We believe that the physics of the rectangular MAS problem has two more hyperparameters $N_r/N_c$ and $k_r/k_c$ leading to a 4-dimensional phase diagram. Its exploration is left for future work.

We encode principal submatrices, i.e. their row/column index set $I$, as Boolean vectors $\sigma = \{ \sigma_i \}_{i=1}^N \in \{0, 1\}^N$ such that, if $i \in I$ then $\sigma_i = 1$ and vice versa.
We fix the size of the submatrix to $k = m N$, which in the Boolean representation translates to the condition $\sum_i \sigma_i = m N$.
We call $m \in (0,1)$ magnetization.
The average of the entries of a submatrix $\sigma$ can be  then expressed as 
\begin{equation}
  A = |\sigma| = \frac{1}{m^2 N^2} \sum_{i, j = 1}^N J_{ij} \sigma_i \sigma_j \, .
\end{equation}
We define $a = A \sqrt{N}$, and we will see that $a$ is of order one in the thermodynamic limit.

We probe the energy landscape of the MAS problem by studying the associated Gibbs measure 
\begin{equation}
  p(\sigma) = e^{\beta E(\sigma) + \beta h \sum_{i=1}^N \sigma_i} / Z(\beta, h) \, ,
\end{equation}
where $\beta$ is an inverse temperature that we use to fix the average energy, $h$ is a magnetic field that we use to fix the magnetization $m$ and $Z(\beta, h)$ is the partition function.
The uncommon plus sign in front of the inverse temperature is due to the fact that the problem is a maximisation problem. Thus, small temperatures correspond to large positive energies in this model.
The energy function is defined as 
\begin{equation}
  E(\sigma) = \frac{1}{\sqrt{N}} \sum_{i<j} J_{ij} \sigma_i \sigma_j = \frac{m^2 N}{2} a \, ,
\end{equation} 
which, modulo subleading contributions coming from the diagonal term, is a multiple of the submatrix average $a$. 
As the considered model resembles the classic SK model, note that the mapping of the Boolean spins to $\pm 1$ spins, i.e. $s = 2\sigma-1$, leads to an SK model in a random magnetic field, with couplings correlated to the magnetic field (see SM). Such a model has not been considered in the physics literature as far as we are aware. 

We will compute all thermodynamic observables through the \textit{quenched free entropy}, i.e. $\Phi = \lim_{N\to\infty} \EE_J \log Z(\beta, h) / N$ where $\EE_J$ denotes averaging over the distribution of $J$.
We will see that the free entropy can be expressed as a variational problem for the \textit{overlap} order parameter, which is defined as $q = N^{-1} \sum_{i=1}^N \sigma_i^a \sigma_i^b \in [0, m]$ for two replicas of the system $\sigma^a$ and $\sigma^b$.

\section{Replica analysis of the free entropy}
We compute the quenched free entropy $\Phi$ using the replica formalism \cite{mezard1987spin}, i.e. by using the replica trick $\EE_J \log Z = \lim_{n\to 0} (\EE_J Z^n - 1)/n$, computing $\EE_J Z^n$ for integer values of $n$ and performing an analytical continuation to take the $n \to 0$ limit. 
We perform the analytical continuation under the one-step Replica Symmetry Breaking (1-RSB) ansatz,
in which we assume that the Gibbs measure decomposes into a 2-level hierarchy of pure states.
This hierarchy is characterised by two overlaps: the average overlap between microstates belonging to the same pure state $q_1$, and the one between microstates belonging to different pure states $q_0$. The Parisi parameter $p$ acts as a temperature controlling the trade-off between the free entropy of a single pure state, and the entropy of pure states \cite{monasson}.
After a derivation 
that follows steps standard to the replica method \cite{mezard1987spin, nishimori}, we obtain
the following variational free entropy:
\begin{equation}\label{eq.phiRSB1}
  \begin{split}
    \Phi_{\rm 1-RSB}&(m, q_0, q_1, p) = 
    \\&=
      - \frac{\beta^2}{4} 
      \left[ m^2  + (p-1) q_1^2 - p q_0^2 \right]
      \\&\quad
      + \frac{1}{p}
      \int Du \log\left[ \int Dv
      \left[ 1 + e^{\beta H(u,v)}\right]^{p}\right] 
    \, , \\
    H(u, v) &= h + \frac{\beta}{2} (m - q_1) 
      + \sqrt{q_0} u + \sqrt{q_1 - q_0} v \, ,
  \end{split}
\end{equation}
where $Du$ and $Dv$ denote integration against a standard Gaussian measure.
The variational free entropy depends on the submatrix size/magnetization $m$, the intra-state overlap $q_1$, the inter-state overlap $q_0$ and the Parisi parameter $p$.
To obtain the equilibrium free entropy we extremize the variational free entropy over $m, q_0$ and $q_1$. 
The Parisi parameter must be set to one if the resulting complexity (whose definition we provide in the following) is positive, otherwise, the variational free entropy must also be extremized over $p$.

Under the 1-RSB ansatz we have the following expressions for the observables, to be evaluated at the equilibrium values of the order parameters.
The average energy density (and the submatrix average) equals
\begin{equation}\label{eq.energy}
  e = \frac{m^2}{2} a = \frac{\beta}{2} 
  \left[ m^2  + (p-1) q_1^2 - p q_0^2 \right]
  \, , 
\end{equation}
the total entropy (logarithm of the number of microstates) equals 
\begin{equation}\label{eq.entropy}
  s_{\rm total} = \Phi - \beta h m - \beta e \, , 
\end{equation}
and the complexity (logarithm of the number of pure states contributing to the Gibbs measure) equals
\begin{equation}\label{eq.complexity}
  \Sigma = \max\left(0, 
    \del_p \left[\extr_{m, q_0, q_1} \Phi_{\rm 1-RSB} \right]_{p=1}
  \right)    \, .
\end{equation}
If the complexity is non-zero, the total entropy decomposes as 
$s_{\rm total} = s_{\rm internal} + \Sigma$,
where the internal entropy $s_{\rm internal}$ is the logarithm of the number of microstates contributing to each of the exponentially many pure states.

Finally, we need to investigate whether the 1-RSB result is exact in the thermodynamic limit. This is done by analyzing the stability of the 1-RSB ansatz against perturbation of higher-order RSB nature. We perform the so-called Type-I stability analysis (see SM Section I.d), obtaining the stability condition
\begin{equation}\label{eq.stab}
    \begin{split}
        \int Du \,  
        \frac
        {\int Dv \, \left(1 + e^{\beta H}\right)^{p}  \left[
            \ell(\beta H)^2 \left(1 - \ell(\beta H)\right)^2
        \right]}
        {\int Dv \, \left(1 + e^{\beta H}\right)^{p}}
        < \frac{1}{\beta^2} \, ,
    \end{split}
\end{equation}
where $\ell(x) = 1/(1+\exp(-x))$.  When this condition is not satisfied then the 1-RSB results are just an approximation and more steps of RSB need to be taken into account.

We derived the variational free entropy using the replica trick. 
Note, however, that the proof of the full-RSB free entropy giving the exact solution in the thermodynamic limit from \cite{Panchenko} applies to the MAS problem and thus our setting. 
In order to apply their result to our model we note that \cite{Panchenko} constrains the free-entropy to fixed self-overlap $q_{\rm self} = N^{-1} \sum_i \sigma_i^2$, while we constrain the model to fixed magnetization $m = N^{-1} \sum_i \sigma_i$. Due to the choice of Boolean spins, we have that $\sigma_i^2 = \sigma_i$, so that the two constraints coincide (this is not true in general).

\section{The phase diagram}\label{sec.phases}

After solving the above equations, we identify five distinct phases for finite $m \in (0,1)$, and we plot them in Figure \ref{fig.ph}.

\paragraph*{RS phase ---}
For small submatrix-average (corresponding to large temperatures) we observe a replica-symmetric (RS) phase, in which the extremum of the variational free entropy is attained at $q^*_0 = q^*_1$. 
In this phase, the complexity is zero, while the total entropy is strictly positive.
As the submatrix average increases (i.e. the temperature is lowered), the system undergoes a phase transition to an RSB phase.
The nature of the transition is different for $m \le m_c$ and $m \ge m_c $.
We start discussing the former case.

\paragraph*{Dynamical 1-RSB phase ---}
For $m \le m_c$, we observe a discontinuous transition at a value $a_{\rm dynamic}(m)$ of the sub-matrix average from the RS phase to a dynamical 1-RSB phase, in which the extremum of the variational free entropy satisfies $q^*_0 \neq q^*_1$. 
This transition is identified by a sharp jump of the complexity from zero to a positive value, 
meaning that the measure shatters into an exponential number of pure states each with non-zero entropy.
As $m$ decreases, the appropriately rescaled internal entropy decreases, suggesting that in the $m\to 0$ limit, this phase becomes a frozen 1-RSB phase, see Figure \ref{fig.entropies} in SM.

\paragraph*{Static 1-RSB phase ---}
For $m \le m_c$, in the dynamical 1-RSB phase, the complexity continuously decreases as the submatrix average increases. At the value of submatrix average $a_{\rm static}(m) > a_{\rm dynamic}(m)$ at which the complexity vanishes we observe a first-order phase transition, from the dynamical 1-RSB phase to a static 1-RSB phase ($q^*_0 \neq q^*_1$, zero complexity, positive entropy). 

\paragraph*{Full-RSB phase ---}
For values $m \ge m_c$ we observe a continuous phase transition from RS to full-RSB phase at $a_{\rm stability}(m)$. The transition happens when the Type-I 1-RSB stability condition \eqref{eq.stab} fails. We conjecture that this phase is fully replica-symmetry-broken by the similarity to the SK model, in which a similar continuous phase transition to full-RSB occurs.
For values of $m_c \ge m \ge m^*$, we observe a Gardner-like phase transition at a value $a_{\rm stability}(m)$ of the sub-matrix average from the 1-RSB  phase to a full-RSB phase. This transition is reminiscent of the one known from the Ising $p$-spin model \cite{gardner1985spin}.  

\paragraph*{UNSAT phase ---}
As the average of the matrix increases, we encounter a point at which the total entropy vanishes, denoting that the sub-matrix average has reached its maximum value $a_{\rm max}(m)$. After this point, the total entropy becomes negative and we observe an unsatisfiable (UNSAT) phase, where no submatrix with that value of the submatrix-average exists.
For $m \ge m^*$, we provide only an approximate estimate of this transition line (while all other transitions presented are exact up to the precision of the numerical solver). 
We estimated $a_{\rm max}(m)$ in the 1-RSB solution, even though this ansatz is unstable in this phase, by computing the 1-RSB entropy, finding the temperature at which it vanishes, and computing the corresponding submatrix average.

It is often the case that the 1-RSB prediction for the maximum energy is numerically very close to the full-RSB prediction. 
Evaluation of the the full-RSB equations, which are proven to give the correct maximum average (analogous to the ground-state energy in the SK model) \cite{Panchenko}, is left for future work.

\paragraph*{Tricritical points ---}
The phase diagram features two tricritical points. The first one, at $(m_c, a_c) \approx (0.09-0.1, 9.3-9.7)$ marks the coexistence of the RS, 1-RSB and full-RSB phases. It can be pinpointed by finding the intersection of the stability and the static transition lines. As $m \to m_c^-$, the static and dynamic transitions approach very quickly so that it is very difficult to distinguish them numerically.
The second tricritical point is at $(m^*, a^*) \approx (0.0001-0.002, 100-650)$, marking the crossing between the stability and the UNSAT transition lines. For $m \le m^*$ the 1-RSB phase is stable up to the maximum average $a_{\rm max}(m)$.
This second tricritical point is hard to pinpoint numerically accurately. In the SM, we show analytically that at least for $m\to 0$ the 1-RSB phase is indeed stable up to $a_{\rm max}(m)$. Thus, by continuity, this second tricritical point must exist.

\section{The small magnetization limit}

We now study the phase diagram in the $m\to 0$ limit, corresponding to the $1 \ll k \ll N$ regime.
The limit must be taken carefully in order to preserve the extensivity of the energy function in the thermodynamic limit.
Indeed, we have that for fixed $m$ and $N$
\begin{equation}
  \text{var}(E(\sigma)) = \caO\left( m^2 N \right) \, , \quad  \# = \caO(N m \log 1/m ) \, ,
\end{equation}
where $\#$ denotes the logarithm of the number of microstates at fixed $m$ and $N$.
As $m\to 0$, the energy must be rescaled by $c(m) = \sqrt{\log(1/m)/m}$, and the entropy and complexity must be rescaled by $m \log(1/m)$. 
This can be achieved by considering the $m\to 0$ limit of \eqref{eq.phiRSB1} at fixed $b = \beta /c(m)$.
We perform analytically the limit in the RS and 1-RSB solutions, leading to the following phase diagram.

\paragraph*{RS phase ---} 
For sub-matrix average $a < a_{\rm dynamic} = \sqrt{2} \, c(m)$, we observe a stable RS phase with zero complexity and positive total entropy. In this phase, $q_0 = q_1 = m^2$.

\paragraph*{Frozen 1-RSB phase ---} 
For sub-matrix average $a_{\rm dynamic} < a < a_{\rm static} = 2 \, c(m)$, we observe a stable 1-RSB phase with $q_1 = m$, $q_0 = m^2$ and complexity $\Sigma  = 1 - b^2/4$.
This is a frozen phase, meaning that each of the exponentially-many pure states contributing to the measure 
has zero internal entropy.

\paragraph*{UNSAT phase ---} 
For sub-matrix average $a > a_{\rm static}$, the total entropy is negative, signalling the onset of the UNSAT phase.

The threshold $a_{\rm dynamic}$ and $a_{\rm static}$ coincide with the thresholds proved in \cite{Bhamidi2017} for, respectively, the submatrix-average of the local maxima $A_{\LAS} = a_{\rm dynamic} / \sqrt{N}$  and the maximum submatrix-average achievable $A_{\rm opt} = a_{\rm static}/ \sqrt{N}$.
Thus, as a byproduct of our analysis, we obtain an equilibrium interpretation of $A_{\LAS}$ as a freezing transition.

The $m\to 0$ limit of the MAS phase diagram resembles closely that of the Random Energy Model (REM) \cite{derrida1981random}.
More precisely, we find that the static threshold, as well as the values of the entropy and complexity, do coincide 
(in the REM 
the static transition threshold equals $a_{\rm static, REM} = 2$, the complexity equals $\Sigma  = 1 - b^2/4$ and the internal entropy is zero for all $b > 0$).
This connection is related to the fact that the MAS energy $E(\sigma)$ is a Gaussian random variable with 
covariance $\angavg{E(\sigma)E(\sigma')} \propto q(\sigma, \sigma') \leq m$, which vanishes in the $m \to 0$ limit.
The notable difference between the REM and the MAS problem is given by the finite value of the dynamic threshold in the MAS problem, while in the REM the system is frozen at all temperatures.

In \cite{Gamarnik2018}, the authors introduced an algorithm called $\IGP$, and they proved that it can find submatrices with average $A$ which, following our analysis, are inside the frozen 1-RSB region. 
This is at odds with the common belief that solutions in frozen states are algorithmically hard to find   \cite{gamarnik2022disordered}. 
Another problem in which a similar situation happens is the binary perceptron, where the algorithmic feasibility was explained in relation to out-of-equilibrium dense regions \cite{baldassi2015subdominant, baldassi2016unreasonable}.
We leave for future work a deeper understanding of the out-of-equilibrium properties of the MAS problem, and more generally the relation between them and algorithmic tractability.

\section{Conclusions}

In this paper, we studied the maximum average submatrix problem using tools from the statistical physics of disordered systems, and in particular a mapping onto a variant of the SK model. 

We unveiled the phase diagram in the large submatrix regime $k = m N$, discovering a rich phenomenology including glassy phases, and phases where exponentially-many pure states contribute to the equilibrium behavior of the system.

By considering the $m\to 0$ limit, we characterized the phase diagram in the small submatrix regime $k \ll N$, shedding some light on previous results \cite{Bhamidi2017,Gamarnik2018} and highlighting a connection to the Random Energy model. 
We note that there exist efficient algorithms that work in the frozen 1-RSB phase, usually associated with hard-algorithmic phases, similar to what happens in the binary perceptron due to non-equilibrium phenomena.

Our findings leave many questions to be answered, such as i) the study of the out-of-equilibrium properties of the problem and their relation with algorithmic hardness and ii) how for $k \ll N$ the vanishingly small correlations between the energies combine to shift the dynamical temperature from infinity (REM) to finite (MAS).

We conclude by remarking that our techniques generalise straightforwardly to the case in which the entries of $J$ are non-Gaussian as long as they are i.i.d. with finite first and second moment, and to the rectangular MAS problem, in which both $J$ and the submatrices may be rectangular possibly with different aspect ratios.

\begin{acknowledgments}
We acknowledge funding from 
the Swiss National Science Foundation grants $200021\_200390$ (OperaGOST) and  $\text{TMPFP2}\_210012$.
\end{acknowledgments}

\providecommand{\noopsort}[1]{}\providecommand{\singleletter}[1]{#1}%
%

\pagebreak
\clearpage

\widetext
\begin{center}
\textbf{\large Statistical mechanics of the maximum-average submatrix problem}\\  \vspace{0.1cm} \textbf{\large Supplemental Materials}
\end{center}
\setcounter{section}{0}
\setcounter{equation}{0}
\setcounter{figure}{0}
\setcounter{table}{0}
\setcounter{page}{1}
\makeatletter
\renewcommand{\theequation}{S\arabic{equation}}
\renewcommand{\thefigure}{S\arabic{figure}}
\renewcommand{\bibnumfmt}[1]{[S#1]}

The code used to produce the plots can be found at 
\url{https://github.com/SPOC-group/Maximum-Average-Submatrix.git}.

\section{Mapping from boolean to binary spins}

The energy of the system is defined as 
\begin{equation}
  E(\sigma) = \frac{1}{\sqrt{N}} \sum_{i<j} J_{ij} \sigma_i \sigma_j \, ,
\end{equation} 
where $\sigma$ is a configuration of $N$ boolean spins $\sigma \in \{0, 1\}^N$.
The straight-forward mapping from boolean to binary spins, i.e. defining $s_i = 2\sigma_i - 1$, leads to the equivalent energy function
\begin{equation}
    \tilde E(s) 
    = \frac{1}{\sqrt{N}} \sum_{i<j} J_{ij} \frac{s_i+1}{2} \, \frac{s_j+1}{2} 
    = 
      \frac{1}{4 \sqrt{N}} \sum_{i<j} J_{ij} s_i s_j
    + \frac{1}{4 \sqrt{N}} \sum_{i} s_i \sum_{j\neq i} J_{ij}
    + \frac{1}{4 \sqrt{N}} \sum_{i<j} J_{ij}  
    \, .
\end{equation}
Thus, written as a function of binary spin configurations $s$, the energy has the same SK-like interaction term, but develops  an additional random magnetic field interaction (the second term). 
Notice that the random field is not independent from the pair-wise interactions $J$, leading to a model which is quantitatively different from the usually studied SK model.

\section{Universality of the phase transitions}

In Figure \ref{fig.ph} we presented the phase diagram for the MAS problem for finite $m$ and small $m$, and observed a variety of phase transitions. 
In the following table we provide references to a selection of other mean-field models in which the same type of phase transitions arise. We order then in a way that the upper lines are analogous to large $m$ and lower lines to small $m$ in the MAS problem.  
The phenomenology at each transition of the MAS problem is qualitatively similar to the phenomenology of the corresponding transition in the models listed below.

We notice that a tricritical point similar to the one we found at $m = m_c$ ("Tricritical 2" in Figure \ref{fig.ph})
is present in the perceptron with negative margin.
We could not find an example of a model in which a tricritical point akin to the one we find at $m = m_*$ ("Tricritical 1" in Figure \ref{fig.ph}) arises.

\begin{center}
\renewcommand{\arraystretch}{1.5}
\begin{tabular}{l|l}
    Model
    & Phase transitions
    \\
    \hline
    SK ($p=2$ binary $p$-spin model) \cite{sherrington1975solvable, mezard1987spin}
    & RS $\to$ Full-RSB $\to$ UNSAT
    \\
    Perceptron with margin $\kappa \geq \kappa_{\rm 1RSB}$  \cite{zamponi}
    & RS $\to$ Full RSB  $\to$ UNSAT
    \\
    $p > 2$ binary $p$-spin model \cite{gardner1985spin}
    & RS $\to$ d-1RSB $\to$ s-1RSB $\to$ Full RSB $\to$ UNSAT
    \\
    Perceptron with margin $\kappa \leq \kappa_{\rm RFOT}$ \cite{zamponi}
    & RS $\to$ d-1RSB $\to$ s-1RSB $\to$ Full RSB  $\to$ UNSAT
    \\
    $p > 2$ spherical $p$-spin model \cite{Crisanti1992}
    & RS $\to$ d-1RSB $\to$ s-1RSB $\to$ UNSAT
    \\
    $q$-coloring ($q \geq 3$) \cite{qCol}
    & RS $\to$ d-1RSB $\to$ s-1RSB $\to$ UNSAT
    \\
    Locked constraint satisfaction problems \cite{locked}
    & RS $\to$ Frozen 1RSB $\to$ UNSAT
    
\end{tabular}

\vspace{0.25cm}
Legend: s-1RSB = static 1-RSB, d-1RSB = dynamic 1-RSB.
\renewcommand{\arraystretch}{1.}
\end{center}

\section{The 1-RSB solution}

\paragraph{Variational free entropy ---}
To derive the variational free entropy \eqref{eq.phiRSB1}, one can follow the derivation for the SK model presented in \cite{nishimori}.
The only difference in our case is that $\sigma^2 = \sigma$ as $\sigma \in \{0,1\}$, contrary to the usual case $\sigma_{\rm SK}^2 = 1$ as $\sigma_{\rm SK} \in \{-1, +1\}$.
The extremization conditions are
\begin{equation}\label{eq.1-RSBSP}
    \begin{split}
        m &= \int Du \,\,\frac{\int Dv \,\, \logistic(\beta H)(1+e^{\beta H})^{p}}{\int Dv \,\,(1+e^{\beta H})^{p}} \, ,\\
        q_0 &= \int Du \,\,\left[\frac{\int Dv \,\, \logistic(\beta H)(1+e^{\beta H})^{p}}{\int Dv \,\,(1+e^{\beta H})^{p}}\right]^2 \, ,\\
        q_1 &= \int Du \,\,\frac{\int Dv \,\, \logistic^2(\beta H)(1+e^{\beta H})^{p}}{\int Dv \,\,(1+e^{\beta H})^{p}} \, ,\\
        \Sigma(p) &= p^2\frac{\beta^2(q_1^2-q_0^2)}{4} - p\int Du \,\frac{\int Dv \,\logpexp(\beta H)(1+e^{\beta H})^{p}}{\int Dv \,\,(1+e^{\beta H})^{p}} + \int Du \,\,\log\left[ \int Dv \,\,(1+e^{\beta H})^{p}\right] \,,
    \end{split}
\end{equation}
where $\logistic(x) = \left(1 + e^{-x}\right)^{-1}$ and $H(u, v) = h + \frac{\beta}{2} (m - q_1) + \sqrt{q_0} u + \sqrt{q_1 - q_0} v\equiv H_{\rm 1-RSB}$. We highlighted that the extremization condition inn $p$ is equivalent to imposing zero complexity.
Their derivation requires the usage of integration by parts repeatedly, in the form 
\begin{equation}
    \int Du \, u f(u) = \int Du \, f'(u) \, .
\end{equation}

\paragraph{Observables ---}
To derive the expressions for the energy $e$ \eqref{eq.energy} and total entropy $s$ \eqref{eq.entropy} in the 1-RSB solution, use the grand-canonical thermodynamic relations $e = \del_\beta \Phi -h m$ and $\Phi(\beta) = s + \beta e + \beta h m$.
Applying these relations to the variational free entropy \eqref{eq.phiRSB1} produces the equations presented in the text, to be evaluated then at the equilibrium values of the order parameters.

\paragraph{Complexity ---}
To derive the expression for the complexity $\Sigma$ \eqref{eq.complexity} in the 1-RSB solution, we follow \cite{monasson}.
We start by introducing a deformed partition function
\begin{equation}
    Z(p) = \sum_\alpha Z_\alpha^p = \sum_\alpha e^{N p \phi_\alpha}
    = \sum_\phi e^{N ( p \phi + \Sigma(\phi) )}
    = e^{N \extr_{\phi : \Sigma(\phi) \geq 0} ( p \phi + \Sigma(\phi) )  } \, ,
\end{equation}
where $p$ is just a weighting parameter --- at $p=1$ we are computing the usual partition function $Z$ --- and the sum over $\alpha$ runs over all pure states contributing to the Gibbs measure, each with free entropy $\phi_\alpha$.
One can interpret $p$ as being a number of replicas of our system that are constrained to be in the same pure state, and $Z(p)$ would then be the correct partition function for such replicated system. Thus, we have 
$\EE \log Z(p) / N \sim p \Phi_{\rm 1-RSB}(p)$, where the factor $p$ matches the number of replicas used on the two sides of the equation.
Now, the extremization over $\phi$ on the left-hand side can be performed explicitly, and the extremizer $\phi_*(p)$ satisfies
\begin{equation}
    \begin{cases}
        p + \Sigma'(\phi_*(p)) = 0 & \text{if }\Sigma(\phi_*(p)) \geq 0 \, ,\\
        \Sigma(\phi_*(p))  = 0 & \text{otherwise} \, .
    \end{cases}
\end{equation}
Whenever $\Sigma(\phi_*(p)) \geq 0$  
\begin{equation}
    \begin{split}
        \del_p \Phi_{\rm 1-RSB}(p)
        &= \del_p \frac{p \phi_*(p) + \Sigma(\phi_*(p))}{p} 
        = \frac{p^2 \del_p \phi_*(p) +p \Sigma'(\phi_*) \del_p \phi_*(p) - \Sigma(\phi_*(p))}{p^2}
        = -\frac{\Sigma(\phi_*(p))}{p^2} \, ,
    \end{split}
\end{equation}
so that 
\begin{equation}
    \Sigma(\phi_*(p)) = - p^2 \del_p \Phi_{\rm 1-RSB}(p) \, .
\end{equation}
Plugging $p=1$, corresponding to the equilibrium partition function, we get
\begin{equation}
    \Sigma(\phi_*) = \max\left(0, - \del_p \Phi_{\rm 1-RSB}(p=1) \right)\, .
\end{equation}

\paragraph{Stability ---}\label{sec.SMstability}

The most general stability condition for the $k$-RSB solution can be derived by computing the Hessian at the $k$-RSB equilibrium of the $n$-replicas free-entropy. A detailed derivation for the RS solution ($q_0 = q_1$) of the SK model can be found in \cite[Chapter 3]{nishimori}. By adapting it to our model, it is easy but tedious to see that the  the RS  stability condition 
\begin{equation}
    \beta^2 \int Du \, \logistic(\beta H_{\rm RS})^2 \left(1 - \logistic(\beta H_{\rm RS})\right)^2 < 1\, ,
\end{equation}
where $H_{\rm RS}(u) = \sqrt{q} z + h + \frac{\beta}{2}( m - q )$.

For the 1-RSB stability, we check a simpler condition, called Type-I stability. 
It is the linear stability of the 2-RSB extremization conditions (seen as fixed point iterations for the overlaps) around the 1-RSB fixed point under the perturbation $q^{\rm 2-RSB}_2 = q^{\rm 1-RSB}_1 + \epsilon$ and
        $q^{\rm 2-RSB}_{0, 1} = q^{\rm 1-RSB}_{0, 1}$.
We conjecture that this threshold is not just a bound, but the correct one, in analogy with what happens in the SK model.
In practice, we consider the equation for $q_2$ in the 2-RSB solution, i.e. 
\begin{equation}
    q_2 = \int Du \,\,
        \frac
    {\int Dv \, \left[ 
    \left(\int Dz \, (1+e^{\beta H_{\rm 2RSB}})^{p_2} \right)^{p_1/p_2 -1 }
    \int Dz \, (1+e^{\beta H_{\rm 2RSB}})^{p_2} \logistic^2(\beta H_{\rm 2RSB}) \right]}
    {\int Dv \, 
    \left(\int Dz \, (1+e^{\beta H_{\rm 2RSB}})^{p_2} \right)^{p_1/p_2 }} \, ,
\end{equation}
where $H_{\rm 2RSB} (u,v,z) =  h + \frac{\beta}{2} (m - q_2) 
    + \sqrt{q_0} u + \sqrt{q_1 - q_0} v + \sqrt{q_2 - q_1} z$, and we evaluate it at the 1-RSB equilibrium value of the order parameters, i.e. $q^{\rm 2-RSB}_{0, 1} = q^{\rm 1-RSB}_{0, 1}$ and $p^{\rm 2-RSB}_1 = p^{\rm 1-RSB}$ for a perturbation $q^{\rm 2-RSB}_2 = q^{\rm 1-RSB}_1 + \epsilon$ at fixed 2-step Parisi parameter $p^{\rm 2-RSB}_2 \in (p^{\rm 1-RSB}, 1)$.
Finally, we expand at first order in $\epsilon$. We report the details of the computation in Section \ref{label.stabdetails}. The linear stability threshold is the point at which the $\caO(\epsilon)$ of the r.h.s. equals 1. This corresponds to the condition (here $p^{\rm 2-RSB}_2 = 1$ gives the strictest condition), giving \eqref{eq.stab}.

\section{Numerical solution for the equilibrium order parameters}

The numerical solution for the equilibrium order parameters is performed by solving the saddle-point equations
\begin{equation}\label{eq.sp}
    \begin{split}
        m &= \int Du \,\,\frac{\int Dv \,\, \logistic(\beta H)(1+e^{\beta H})^{p}}{\int Dv \,\,(1+e^{\beta H})^{p}} \, ,\\
        q_0 &= \int Du \,\,\left[\frac{\int Dv \,\, \logistic(\beta H)(1+e^{\beta H})^{p}}{\int Dv \,\,(1+e^{\beta H})^{p}}\right]^2 \, ,\\
        q_1 &= \int Du \,\,\frac{\int Dv \,\, \logistic^2(\beta H)(1+e^{\beta H})^{p}}{\int Dv \,\,(1+e^{\beta H})^{p}} \, ,\\
        \Sigma(p) &= p^2\frac{\beta^2(q_1^2-q_0^2)}{4} - p\int Du \,\frac{\int Dv \,\logpexp(\beta H)(1+e^{\beta H})^{p}}{\int Dv \,\,(1+e^{\beta H})^{p}} + \int Du \,\,\log\left[ \int Dv \,\,(1+e^{\beta H})^{p}\right] \,,
    \end{split}
\end{equation}
where $\logistic(x) = \left(1 + e^{-x}\right)^{-1}$ and $H(u, v) = h + \frac{\beta}{2} (m - q_1) + \sqrt{q_0} u + \sqrt{q_1 - q_0} v$.
There are several aspects that must be explained.

\paragraph{What to do with the Parisi parameter $p$ ---} 
The complexity prescribes that at equilibrium $p=1$ if the associated complexity is positive, and otherwise $p=p^*$ such that $\Sigma(p^*) = 0$. We solve the first three equations of \eqref{eq.sp} for $m, q_0, q_1$ at $p=1$ first and compute the associated complexity. If it is negative, we discard the solution and solve also for $\Sigma(p) = 0$.

\paragraph{How to solve \eqref{eq.sp} at fixed $p$ ---}
To solve the first three equations of \eqref{eq.sp} for $m, q_0, q_1$ at fixed $p$, we could turn them in a fixed-point iteration scheme, as it is done for the SK model. An added element of complexity is that we actually want to fix $m$, and solve for $h$. This cannot be done \textit{a posteriori} as the $m(h)$ function at the stable fixed point of the equation is not single-valued (while the inverse $m(h)$ is). Thus, we iterate the equations for $q_0$ and $q_1$ as fixed-point equations, and after each iteration we solve the first equation at fixed $m$ for $h$ --- by bisection for example --- using the current value of the other order parameters. Whenever no solution for $h$ is found, we perturb it with random Gaussian noise of small variance, and reiterate the procedure.
Convergence is declared when the relative change in Euclidean norm of the order parameters is below a set tolerance, in our case $10^{-6}$.

\paragraph{How to solve \eqref{eq.sp} imposing $\Sigma(p)=0$ ---}
We use the same procedure as that described above, but every 3 iterations we also solve $\Sigma(p)=0$ for $p$, using the current value of the other order parameters. 

\paragraph{How to stably compute the integrals in \eqref{eq.sp} --- }
To compute the integrals in $Du$ and $Dv$ in the saddle-point equations we use Gauss-Hermite integration, which approximates Gaussian expectations as 
\begin{equation}
    \int Dz \, f(z) = \sum_{i=1}^{n} w_i f(z_i)
\end{equation}
for a specific set of weights $w_i$ and base-points $z_i$.
We use $n = 71$, and pre-compute the weights and base-points in advance.

Notice that this integrals can lead to overflows and underflows due to exponential factors.
For example, consider the simple RS case $q_0 = q_1$. In this case, all integrals in the $v$ variable drop, as there is no dependence on $v$ anymore, leading to a cancellation of the factors $(1+\exp(\beta H))^p$. 
If this cancellation has to happen numerically, it can lead to the aforementioned numerical issues.
To solve the equations in a numerically stable way, we work in log-space. We write
 the integrals as 
\begin{equation}
    \begin{split}
        \int &Du \,\,\frac{\int Dv \,\, f(\beta H)(1+e^{\beta H})^{p}}{\int Dv \,\,(1+e^{\beta H})^{p}} 
        \\&= 
        \int Du \, \exp \left[ 
             \log \int Dv \,\, f(\beta H)(1+e^{\beta H})^{p} - \log \int Dv \,\,(1+e^{\beta H})^{p}
        \right]
        \\
         &= 
        \int Du \, \exp \left[ 
             \log \int Dv \, \exp \left[ \log f(\beta H) + p \logpexp(\beta H) \right] 
             - \log \int Dv \, \exp \left[p \logpexp(\beta H) \right]
        \right] \, .
    \end{split}
\end{equation}
Now we use that
\begin{equation}
    \log \int Dz \, \exp(f(z)) 
    = \log \sum_{i=1}^{n} w_i \exp(f(z_i))
    = \log \sum_{i=1}^{n} \exp(f(z_i) + \log w_i)
    = \text{logsumexp}\left(  \{f(z_i) + \log w_i \}_{i=1}^n  \right) \, .
\end{equation}
where stable logsumexp functions are routinely implemented in many programming languages.
The idea of the logsumexp trick is to write 
\begin{equation}
    \log \sum_i \exp(x_i) = x_* + \log \sum_i \exp(x_i - x_*)
\end{equation}
where $x_* = \max_i \{x_i\}$, so that all exponentials have negative argument.
This allows to compute the integrals quickly and stably.

\paragraph{How to enter the large inverse temperature 1-RSB phase and the dynamical 1-RSB phase}
In both the large inverse temperature 1-RSB phase and the dynamical 1-RSB phase, the solution of the 1-RSB equations depends highly on the initialization. In the first case, poor initialization leads to non-convergence, while the the second case it leads to convergence to the RS fixed point $q_0 = q_1$, in which the complexity erroneously vanishes.
To avoid these problems, we manually find good initializations in the full-RSB region or static 1-RSB respectively, at an inverse temperature just above the relevant phase boundary. We then change slightly $\beta$ (increasing or decreasing it as needed) and feed the previous solution of the equations as initialization at the new inverse temperature. This allows to go at larger values of $\beta$ in the first case, and inside the dynamical 1-RSB region in the second case.

\paragraph{Extrapolation of the observables to compute the stability and the SAT/UNSAT thresholds}
In order to compute the stability and SAT/UNSAT thresholds we need to extrapolate the values of the energy, entropy and the stability condition to large values of $\beta$, larger than reachable with our numerical methods. 
All these quantities behave as $\text{const} + \caO(\beta)$ corrections in the physically relevant region (deep in the unstable phase this is not necessarily true anymore for the stability condition). 
Thus, we fit the large inverse temperature tails of said observables to the  $\text{const} + \caO(\beta)$ asymptotics, and estimate the stability and SAT/UNSAT thresholds on this extrapolated data as the points at which, respectively, the stability condition is not satisfied and the entropy becomes negative.
\paragraph{Estimation of the dynamic and static transitions}
We estimate the dynamic and static transition by considering the complexity (at $p=1$), and computing i) where it first develops a discontinuity jumping from zero to positive as a function of $\beta$, and ii) where it firsts continuously evolves from positive to negative. 
To properly threshold for these events, it is important to rescale the complexity with its natural scaling $m \log(1/m)$ (see the main-text discussion for a justification of this scaling).

\paragraph{Phase diagram as a function of the inverse temperature}
For completeness, in Figure \ref{fig.phbeta} we present the phase diagram as a function of $m$ and $\beta$.
\begin{figure}
\centering
  \includegraphics[width=0.4\columnwidth]{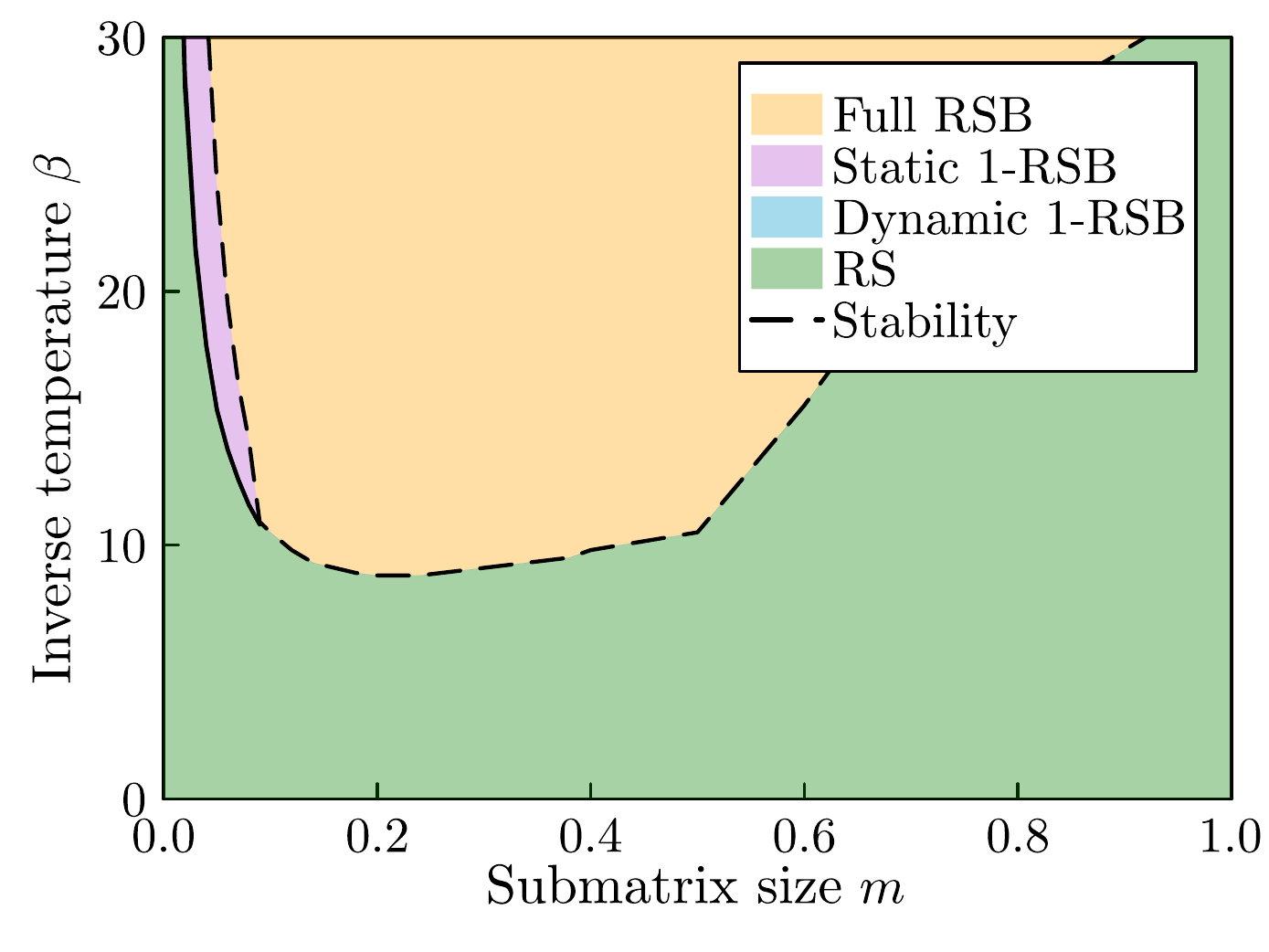}
  \hspace{1cm}
  \includegraphics[width=0.4\columnwidth]{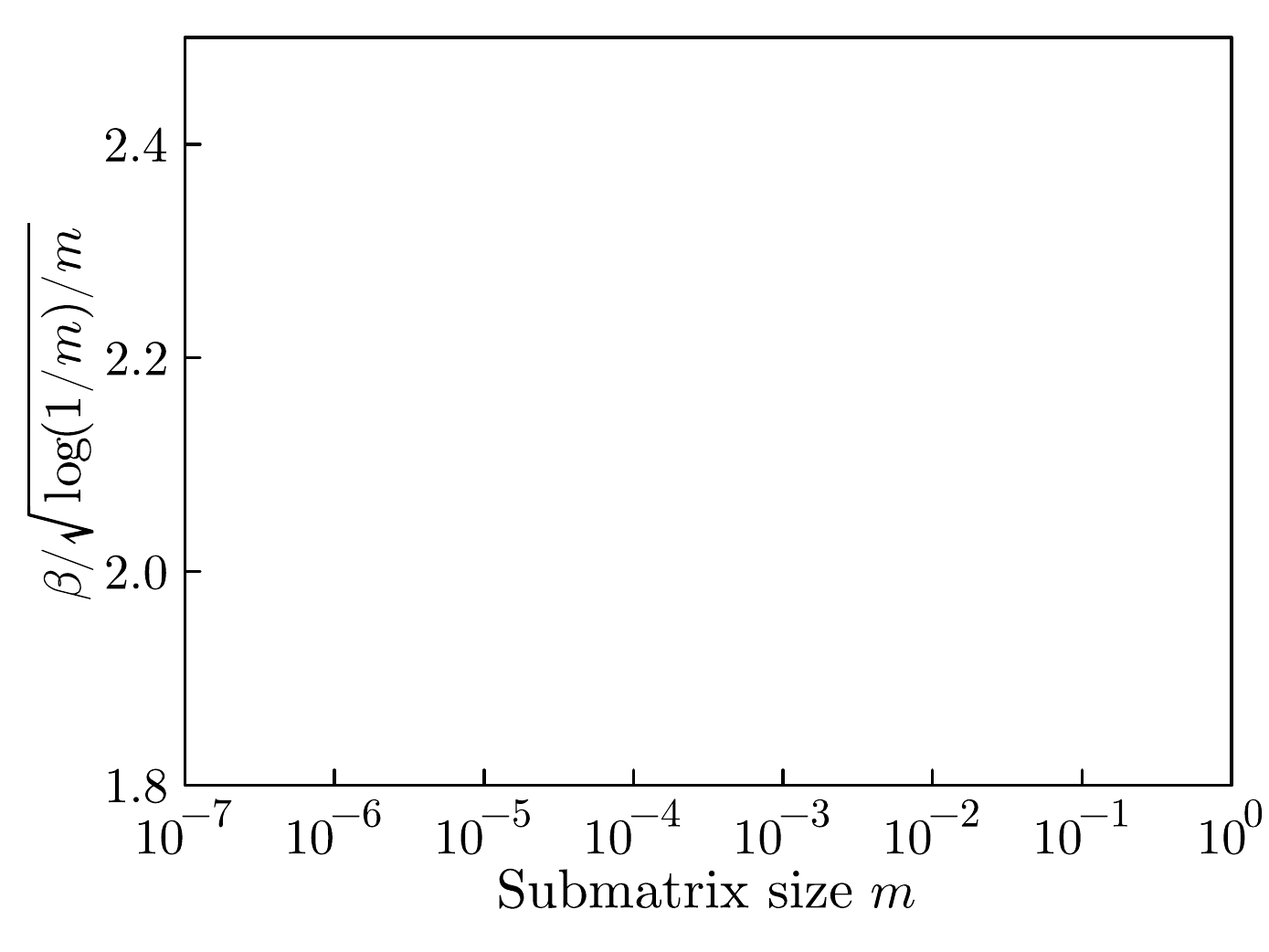}  
  \caption{
    The phase diagram of the MAS problem as a function of the inverse temperature $\beta$ and the submatrix size $m = k/N$ for linear scale in $m$ (left) nad logarithmic scale in $m$ (right).
    On the right we rescaled the inverse temperature to highlight the convergence to the $m\to 0$ limit.
  }
  \label{fig.phbeta}
\end{figure}

\section{The small magnetization limit for the symmetric model}\label{sec.SM1-RSBsmallM}

In this section we derive the $m\to 0$ behaviour of the symmetric MAS problem.
We do this by considering a joint of $m \to 0$ and $\beta \to \infty$, with a scaling relation tuned to preserve the extensivity of the Hamiltonian in the limit.
In particular, we have that for fixed $m$ and $N$
\begin{equation}
  \text{var}(E(\sigma)) = \caO\left( m^2 N \right) \, , \quad  \# = \caO(N m \log 1/m ) \, ,
\end{equation}
where $\#$ denotes the logarithm of the number of microstates at fixed $m$ and $N$.
As $m\to 0$, the energy must be rescaled by $c(m) = \sqrt{\log(1/m)/m}$, and the entropy and complexity must be rescaled by $m \log(1/m)$. 
This can be achieved by considering the $m\to 0$ limit of \eqref{eq.phiRSB1} at fixed $b = \beta /c(m)$.
Below, we perform analytically the limit in the RS and 1-RSB solutions, leading to the following phase diagram.

Notice that our $m \to 0$ results describe the regime $1 \ll k \ll N$. 
Indeed, for $k = \caO(1)$ one must be careful, as the diagonal terms of the submatrix average $\sum_i J_{ii} \sigma_i$ (which we discard as subleading in the thermodynamic limit in our analysis) become important.

\subsection{The scaling limit of the RS solution}

We start by considering the RS equations
\begin{equation}\label{eq.RSSP}
    \begin{split}
        m &= \int Du \, \logistic(\beta H) \, ,\\
        q &= \int Du \, \logistic(\beta H)^2 \, ,\\
        H_{\rm RS} &= h + \sqrt{q} u + \frac{\beta}{2} (m-q) \, , 
    \end{split}
\end{equation}
in the scaling limit 
\begin{equation}
    \begin{split}
        \beta &= b \sqrt{\frac{1}{m} \log \frac{1}{m}} \, ,\\
        h &= \eta \sqrt{m  \log \frac{1}{m}} \, .
    \end{split}
\end{equation}

Then 
\begin{equation}
    \beta H_{\rm RS} = b \left(\eta  + \frac{b}{2} \left(1 - \frac{q}{m}\right) \right) \log \frac{1}{m} + u b \sqrt{\frac{q}{m} \log \frac{1}{m}} \, .
\end{equation}

We notice here that if $q = c m^{1+\alpha}$, with $c>0$ and $\alpha >0$ due to $0 \leq q \leq m$,
the second term goes to zero and can be safely neglected. 
The dependence on $u$ drops, and the saddle point equations imply self-consistently 
\begin{equation}
    q = \logistic(\beta H_{\rm RS})^2 = m^2 \implies c = \alpha = 1 \, .
\end{equation}

The first equation reads
\begin{equation}
    \begin{split}
        m = \logistic\left( b \left(\eta  + \frac{b}{2}  \right) \log \frac{1}{m} \right) 
        \implies
        \log\frac{m}{1-m} = b \left(\eta  + \frac{b}{2}  \right) \log \frac{1}{m} \, ,
    \end{split}
\end{equation} 
from which
\begin{equation}
    \eta = - \left(\frac{1}{b} + \frac{b}{2}\right) \, .
\end{equation}

The energy density equals (we consider the symmetric model here, but for the non-symmetric model the computations are analogous)
\begin{equation}
    e
    = \frac{\beta}{2} (m^2 - q^2) 
    \sim \frac{b}{2} \, \sqrt{m^3 \log \frac{1}{m}} \, .
\end{equation}

Given the \textit{a posteriori} knowledge that the small temperature phase is a 1-RSB phase, we can also access the values of the observables in the frozen 1-RSB phase by 
\begin{itemize}
    \item computing $b_{\rm stability}$ as the valued of $b$ at which the entropy density becomes negative;
    \item computing the observables at $b_{\rm stability}$, as in the 1-RSB phase they do not depend on the temperature anymore.
\end{itemize}
Thus, we have (recall $q \ll m$) 
\begin{equation}\label{eq.RSsmallmEnergy}
    \begin{split}
        \EE\angavg{e} &\sim - \frac{\beta m^2}{2} \\
        \Phi &\sim - \frac{\beta^2 m^2}{4} + \logpexp(\beta H_{\rm RS}) 
            = - \frac{\beta^2 m^2}{4} + \caO\left( m \right)
            \sim - \frac{\beta^2 m^2}{4}
    \end{split}
\end{equation}
and the entropy density
\begin{equation}
    \EE s 
    = \Phi + \beta \EE \angavg{e} - \beta h m 
    = \left(
        1 + \frac{b^2}{2} - \frac{3 b^2 }{4}
    \right) \, m \log \frac{1}{m}
    = \left(
        1 - \frac{b^2 }{4}
    \right) \, m \log \frac{1}{m}
\end{equation}
from which
\begin{equation}
    b_{\rm stability} = 2 \, .
\end{equation}

Notice that, modulo a factor $\log_2$ due a different normalisation of the entropy, the RS $m\to 0$ limit reproduces exactly the results for the RS solution of the Random Energy Model \cite{derrida1981random}.

\subsection{The scaling limit in the 1-RSB solution}

We start by considering the 1-RSB equations
\begin{equation}
    \begin{split}
        m &= \int Du \,\,\frac{\int Dv \,\, \logistic(\beta H)(1+e^{\beta H})^{p}}{\int Dv \,\,(1+e^{\beta H})^{p}}\\
        q_0 &= \int Du \,\,\left[\frac{\int Dv \,\, \logistic(\beta H)(1+e^{\beta H})^{p}}{\int Dv \,\,(1+e^{\beta H})^{p}}\right]^2 \\
        q_1 &= \int Du \,\,\frac{\int Dv \,\, \logistic^2(\beta H)(1+e^{\beta H})^{p}}{\int Dv \,\,(1+e^{\beta H})^{p}} \\
        \Sigma(p) &= p^2\frac{\beta^2(q_1^2-q_0^2)}{4} - p\int Du \,\frac{\int Dv \,\logpexp(\beta H)(1+e^{\beta H})^{p}}{\int Dv \,\,(1+e^{\beta H})^{p}} + \int Du \,\,\log\left[ \int Dv \,\,(1+e^{\beta H})^{p}\right]
    \end{split}
\end{equation}
where
\begin{equation}
    H(u, v) = \sqrt{q_0}u+ \sqrt{q_1-q_0}v + \frac{\beta}{2}(m-q_1) + h \equiv H_{\rm 1-RSB} \, ,
\end{equation}
and $\Sigma(p)$ is the complexity.
The 1-RSB free entropy potential is given by
\begin{equation}
    \begin{split}
        \Phi_{\rm 1-RSB} 
            &= -\frac{\beta^2}{4} \left[m^2 + (s-1) q_1^2 -  pq_0^2\right] + \frac{1}{p} \int Du\,\log\int Dv\,(1+e^{\beta H(u, v)})^{p} \, .
    \end{split}
\end{equation}
We consider the scaling limit 
\begin{equation}
    \begin{split}
        \beta &= b \sqrt{\frac{1}{m} g(m)} \, ,\\
        h &= - \mu \sqrt{m g(m)} \, ,\\
        q_1 &= c^2 m \quad \text{with} \quad c \in (0, 1] \, ,\\
        q_0 &= m^2 \, , \\
        s &= \caO_m(1) \, , \, \, s \in (0,1] \, ,
    \end{split}
\end{equation}
where
\begin{equation}
    g(m) \sim \log \frac{1}{m} \, , \end{equation}
at leading order.
We also assume that $\mu > 0$, in accordance to the RS solution.

We have that
\begin{equation}
    \beta H = - b \mu g + \frac{b^2}{2} (1 - c^2) g + b c \sqrt{g} v + b \sqrt{m g} u
\end{equation}
and we see that the only term that is going to zero is the last term. 
Thus, the dependence on $u$ can be dropped, confirming immediately through the second SP equation that $q_0 = m^2$.

The solution of the equation is detailed in the next section \ref{sec.1-RSB-limit-sol}.
Here we summarise the results.

If $c^2 g \to 0$ --- happening whenever $c$ goes to 0 polynomially --- then the dependence on $v$ drops, and  we get $q_1 = m^2$ similarly as it happened for $q_0$. This retrieves the RS solution computed in the previous Section.

If $c = \caO_m(1)$, then we have $g = - \log m$ to all non-vanishing orders in $m$, and 
\begin{equation}
    \begin{split}
        \mu = \frac{p^2 b^2+2}{2 p b}        
        \quad\text{ and }\quad
        \Sigma(p, b) = \left( 1 - \frac{1}{4} p^2 b^2 \right) m \log \frac{1}{m} 
         \, ,        
    \end{split}
\end{equation}
under the condition $b \geq \max(\sqrt{2} / p, \sqrt{2/p})$.
The equilibrium solution is given by the condition $p=1$ whenever $\Sigma(b, p=1) > 0$, and by the value of $p$ such that $\Sigma(b, p) = 0$ otherwise.
This gives $p=1$ for $\sqrt{2} < b < 2$, and $p = 2/b$ for $b > 2$. 
\begin{itemize}
    \item In the region of positive complexity $p=1$, we have that exponentially many thermodynamic states with the same free entropy contribute to the thermodynamics of the system.
We have that 
\begin{equation}
    \begin{split}
        \Phi &= - \frac{b^2}{4} p \, m \log \frac{1}{m} \, ,\\
        \beta e &= \frac{b^2}{2} p \, m \log \frac{1}{m} \, ,\\
        \beta h m &= - \mu b \, m \log \frac{1}{m} \, ,\\
    \end{split}
\end{equation}
so that 
\begin{equation}
    s_{\rm internal}
    =  \Phi -  \beta h m - \Sigma + \beta e
    = 0 
\end{equation}
confirming that this phase is a frozen 1-RSB phase, and 
\begin{equation}
    a = \frac{2}{m^2 N^{1/2}} e 
    = \frac{2}{m^2 N^{1/2}} \,  \frac{b}{2} m^{3/2} \sqrt{\log \frac{1}{m}}
    = b \sqrt{\frac{1}{mN} \log \frac{1}{m}} \, .
\end{equation}
In particular, this gives $a_{\rm dynamic}$ and $a_{\rm static}$ by substituting the limits $b=\sqrt{2}$ and $b=2$.
\item In the region of zero complexity $p = 2/b$, and 
\begin{equation}
    \begin{split}
        \Phi &= - \frac{b}{2} \, m \log \frac{1}{m} \, ,\\
        \beta e &= b \, m \log \frac{1}{m} \, ,\\
        \beta h m &= - \mu b \, m \log \frac{1}{m} \, ,\\
    \end{split}
\end{equation}
so that 
\begin{equation}
    a = \frac{2}{m^2 N^{1/2}} e 
    = \frac{2}{m^2 N^{1/2}} \, m^{3/2} \sqrt{\log \frac{1}{m}}
    = 2 \sqrt{\frac{1}{mN} \log \frac{1}{m}} \, ,
\end{equation}
confirming that for $b > 2$ we have that $a = a_{\rm max} = a_{\rm stability}$.
\end{itemize}

We find another 1-RSB solution for $0<p<1$ that is not stable under 2-RSB perturbations, and that could not be reproduced  by us  by solving numerically the 1-RSB saddle-point equations, so we discard it as unphysical.

\begin{figure}
  \centering
  \includegraphics[width=1\columnwidth]{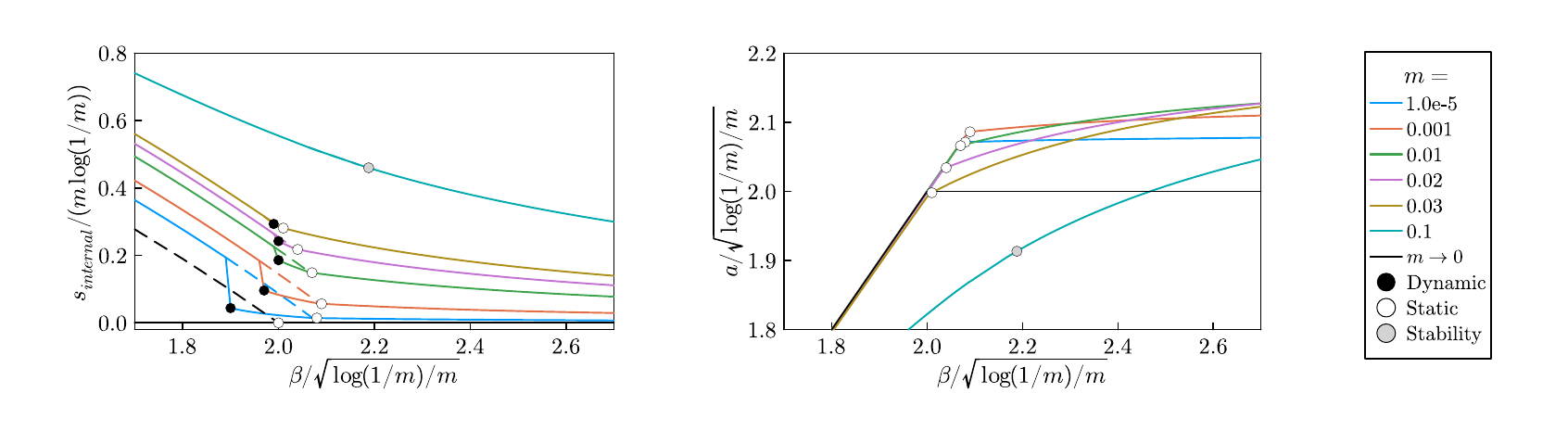}

  \caption{
    Behavior of the internal entropy (left panel, solid lines), total entropy (left panel, dashed lines) and the submatrix average (right panel) as a function of the rescaled inverse temperature $\beta / \sqrt{ \log(1/m) / m}$ for various values of $m$. In black, we plot the analytical predictions for the $m\to 0$ limit.
    We see that both observables converge to the $m\to 0$ limit extremely slowly, in accordance with our analytical analysis that shows that the next-to-leading order in the $m\ll 1$ expansion is only logarithmically small, see SM.
    Moreover, we see that the internal entropy in the dynamical 1-RSB region (where it is different from the total entropy due to the non-vanishing complexity) decreases as $m$ goes to zero, foreshadowing the frozen 1-RSB phase that arises in the limit $m \to 0$.
  }
  \label{fig.entropies}
\end{figure}

\subsection{Solution of the 1-RSB SP equations in the scaling limit}
\label{sec.1-RSB-limit-sol}
    
We start by noticing that for $c^2 \to 0$ such that $c^2 g \to 0$ the dependence on $v$ drops, and  we get $q_1 = m^2$ similarly as it happened for $q_0$. This retrieves the RS solution.
If instead the $v$ dependence does not drop, then we can expect $q_1 = c^2 m$ with $c \neq 0$. 
Thus, our task is now to solve the remaining three SP equations for $\mu, c, p$ at fixed $b$ and for $c\neq 0$, $\mu > 0$.

Notice that, calling 
\begin{equation}
    A = \beta h + \frac{\beta^2}{2} (m-q_1) \sim \left( - b \mu + \frac{b^2}{2} (1-c^2)\right) \log\frac{1}{m}
    \mathand
    B = \beta\sqrt{q}\sim b c \sqrt{\log\frac{1}{m}}
    \mathand 
    N = \sqrt{\log 1/m}
\end{equation}
we have that the 1-RSB equations involve integrals which are in the setting of Proposition \ref{prop.logistic}, with $p > 0$ and $\ell=0, 1, 2$. The condition $B_1 \neq 0$ is equivalent to $c \neq 0$.
Using Proposition \ref{prop.logistic}, we then have that
\begin{equation}
    \begin{split}
        I_{p, 0}(A, B)=\int Du \, (1+e^{\beta H})^p
        \approx 
        \begin{cases}
            e^{f(p) }G\left( \frac{A}{B} + p B \right) 
            & -\frac{A_2}{B_1^2} < \frac{p}{2} 
            \\
            2 e^{f(p) }G\left( \frac{p}{2} B \right)
            & -\frac{A_2}{B_1^2} = \frac{p}{2} 
            \\
            G\left( - \frac{A}{B} \right)
            & -\frac{A_2}{B_1^2} > \frac{p}{2} 
        \end{cases}
    \end{split}
\end{equation}
and
\begin{equation}
    \begin{split}
        I_{p, 1}(A, B)=\int Du \, (1+e^{\beta H})^p \logistic\left(\beta H\right)
        \approx  
        \begin{cases}
                \frac{\exp\left(- \frac{A^2}{2 B^2}\right) }{\sqrt{2 \pi} \, N }K(p, 1, A_2, B_1)
                & p < -\frac{A_2}{B_1^2} < 1 \mathand p < 1
                \\
                e^{f(p) }G\left( \frac{A}{B} + p B \right) 
                & -\frac{A_2}{B_1^2} \leq p \mathand p < 1
                \\
                e^{f(1)}G\left( - \frac{A}{B} - B \right)
                & -\frac{A_2}{B_1^2} \geq 1 \mathand p < 1
                \\
                e^{f(p) }G\left( \frac{A}{B} + p B \right) 
                & -\frac{A_2}{B_1^2} < \frac{1 + p}{2} \mathand p \geq 1
                \\
                2 e^{f(p) }G\left( \frac{p-1}{2} B \right)
                & -\frac{A_2}{B_1^2} = \frac{1 + p}{2} \mathand p \geq 1
                \\
                e^{f(1)}G\left( - \frac{A}{B} - B \right)
                & -\frac{A_2}{B_1^2} > \frac{1 + p}{2} \mathand p \geq 1
        \end{cases}
    \end{split}
\end{equation}
and
\begin{equation}
    \begin{split}
        I_{p, 2}(A, B)=\int Du \, (1+e^{\beta H})^p \logistic\left(\beta H\right)^2
        \approx  
        \begin{cases}
                \frac{\exp\left(- \frac{A^2}{2 B^2}\right) }{\sqrt{2 \pi} \, N }K(p, 2, A_2, B_1)
                & p < -\frac{A_2}{B_1^2} < 2 \mathand p < 2
                \\
                e^{f(p) }G\left( \frac{A}{B} + p B \right) 
                & -\frac{A_2}{B_1^2} \leq p \mathand p < 2
                \\
                e^{f(2)}G\left( - \frac{A}{B} - 2 B \right)
                & -\frac{A_2}{B_1^2} \geq 2 \mathand p < 2
                \\
                e^{f(p) }G\left( \frac{A}{B} + p B \right) 
                & -\frac{A_2}{B_1^2} < \frac{2 + p}{2} \mathand p \geq 2
                \\
                2 e^{f(p) }G\left( \frac{p-2}{2} B \right)
                & -\frac{A_2}{B_1^2} = \frac{2 + p}{2} \mathand p \geq 2
                \\
                e^{f(2)}G\left( - \frac{A}{B} - 2 B \right)
                & -\frac{A_2}{B_1^2} > \frac{2 + p}{2} \mathand p \geq 2
        \end{cases}
    \end{split}
\end{equation}
and the equations to be solved are 
\begin{equation}
    \begin{split}
        I_{p, 1}(A, B) &=     m I_{p, 0}(A, B) \, ,\\
        I_{p, 2}(A, B) &= c^2 m I_{p, 0}(A, B)
    \end{split}
\end{equation}
and recall
\begin{equation}
    f(k) = kA + \frac{1}{2}B^2 k^2 \, .
\end{equation}

We consider the equations in the three cases $0< p < 1$, $1 \leq p < 2$ and $p \geq 2$ separately:
\begin{enumerate}
    \item \textbf{Case $0< p < 1$.} We have the subcases:
    
    \begin{center}
    \renewcommand{\arraystretch}{1.5}
    \begin{tabular}{c||c||c|c|c}
        Case & Condition & $\ell=0$ & $\ell=1$ & $\ell=2$ \\
        \hline
        1a &
        $-A_2/B_1^2 < p / 2$ & 
        $e^{f(p) }G\left( \frac{A}{B} + p B \right)$ &
        $e^{f(p) }G\left( \frac{A}{B} + p B \right)$ &
        $e^{f(p) }G\left( \frac{A}{B} + p B \right)$
        \\
        1b &
        $-A_2/B_1^2 = p / 2$ &
        $2 e^{f(p) }G\left( \frac{p}{2} B \right)$ &
        $e^{f(p) }G\left( \frac{A}{B} + p B \right)$ &
        $e^{f(p) }G\left( \frac{A}{B} + p B \right)$
        \\
        1c &
        $p/2 < -A_2/B_1^2 \leq p$ &
        $G\left( - \frac{A}{B} \right)$ &
        $e^{f(p) }G\left( \frac{A}{B} + p B \right)$ &
        $e^{f(p) }G\left( \frac{A}{B} + p B \right)$
        \\
        1d &
        $p < -A_2/B_1^2 < 1$ &
        $G\left( - \frac{A}{B} \right)$ &
        $\frac{\exp\left(- \frac{A^2}{2 B^2}\right) }{\sqrt{2 \pi} \, N }K(p, 1, A_2, B_1)$ &
        $\frac{\exp\left(- \frac{A^2}{2 B^2}\right) }{\sqrt{2 \pi} \, N }K(p, 2, A_2, B_1)$
        \\
        1e &
        $1 \leq -A_2/B_1^2 < 2$ &
        $G\left( - \frac{A}{B} \right)$ &
        $e^{f(1)}G\left( - \frac{A}{B} - B \right)$ &
        $\frac{\exp\left(- \frac{A^2}{2 B^2}\right) }{\sqrt{2 \pi} \, N }K(p, 2, A_2, B_1)$
        \\
        1f &
        $2 \leq -A_2/B_1^2$ &
        $G\left( - \frac{A}{B} \right)$ &
        $e^{f(1)}G\left( - \frac{A}{B} - B \right)$ &
        $e^{f(2)}G\left( - \frac{A}{B} - 2 B \right)$
        \\
    \end{tabular}
    \renewcommand{\arraystretch}{1.}
    \end{center}

    leading to the following equations (notice that if $- A_2 > 0$ than $G(-A/B)=1$):
    \begin{itemize}
        \item[1a)] Here all integrals are equal at leading order, so there exists no solution to the SP equations which is compatible with the scaling ansatz.
        \item[1b)] Same as 1a).
        \item[1c)] The equations give $c^2 = 1$ and $e^{f(p)} G\left( \frac{A}{B} + p B\right) = m$.
            We have
            \begin{equation}
                c^2 = 1 \implies B_1 = b \mathand A_2 = - b\mu \, .
            \end{equation}
            The second equation implies at leading logarithmic order
            \begin{equation}
                p A_2 + \frac{1}{2}B_1^2 p^2 = - 1 
                \implies 1 - b \mu p + \frac{1}{2}b^2 p^2 = 0
                \implies \mu = \frac{1}{bp} + \frac{bp}{2} \, .
            \end{equation}
        \item[1d)] The equations give 
            \begin{equation}
                \frac{\exp\left(- \frac{A^2}{2 B^2}\right) }{\sqrt{2 \pi \log 1/m}}K(p, 1, A_2, B_1) = m
                \mathand
                c^2 = K(p, 2, A_2, B_1) / K(p, 1, A_2, B_1) \, .
            \end{equation}
            For this to be satisfied, we need to impose for some $\gamma > 0$
            \begin{equation}
                \exp\left(- \frac{A^2}{2 B^2}\right) = \gamma m \sqrt{2 \pi \log 1/m}
                \implies
                - \frac{A^2}{2 B^2} = \log \left( \gamma m \sqrt{2 \pi \log 1/m} \right)
            \end{equation}
            which at leading logarithmic order gives 
            \begin{equation}
                \frac{A_2^2}{2 B_1^2} = 1 \implies A_2 = - \sqrt{2} B_1
            \end{equation}
            Moreover, we have 
            \begin{equation}
                \gamma K(p, 1, - \sqrt{2} bc, bc) = 1 \mathand 
                c^2 = K(p, 2, - \sqrt{2} bc, bc) / K(p, 1, - \sqrt{2} bc, bc)
            \end{equation}
            We thus see that for each values of $b, c, p$, one can tune $\gamma$ such that the last equation is satisfied (notice that $K(p, 1, \dots) > 0$ in this region by construction). Then one needs to solve the second equation for $c$, which in conjunction with the equation for $A_2$ provides an equation for $\mu$, i.e.
            \begin{equation}
                \mu = \frac{b}{2} (1-c^2) - \sqrt{2} c \, .
            \end{equation}
        \item[1e)]
            The equations give 
            \begin{equation}
                e^{f(1)}G\left( - \frac{A}{B} - B \right) = m 
                \mathand
                \frac{\exp\left(- \frac{A^2}{2 B^2}\right) }{\sqrt{2 \pi} \, N }K(p, 2, A_2, B_1)= c^2 m \, .
            \end{equation}
            The first equation implies 
            \begin{equation}
                A_2 + \frac{1}{2}B_1^2 = -1 
                \implies 
                - b \mu + \frac{b^2}{2} (1-c^2) + \frac{b^2c^2}{2} = -1
                \implies 
                \mu = \frac{1}{b} + \frac{b}{2} \, .
            \end{equation}
            The second equation implies $c \to 0$, as the exponent is at leading order
            \begin{equation}
                - \frac{A_2^2}{2 B_1^2} \log\frac{1}{m} = \frac{1}{2} \left( \frac{1}{B_1} + \frac{B_1}{2} \right)^2 \log m \to - \infty
            \end{equation}
            giving back the RS solution. 
        
        \item[1f)]
            The equations give 
            \begin{equation}
                e^{f(1)} G\left( - \frac{A}{B} - B \right) = m \mathand
                e^{f(2)} G\left( - \frac{A}{B} - 2 B \right) = c^2 m \, .
            \end{equation}
            The first equation implies 
            \begin{equation}
                A_2 + \frac{1}{2}B_1^2 = -1 
                \implies 
                - b \mu + \frac{b^2}{2} (1-c^2) + \frac{b^2c^2}{2} = -1
                \implies 
                \mu = \frac{1}{b} + \frac{b}{2} \, .
            \end{equation}
            The second equation implies $c \to 0$, as the exponent is at leading order
            \begin{equation}
                \left( 2 A_2 + 2 B_1^2 \right) \log \frac{1}{m} 
                = \left( 2 - B_1^2 \right) \log m \to - \infty
            \end{equation}
            where we used that
            \begin{equation}
                2 B_1^2 \leq -A_2 = 1 + \frac{1}{2}B_1^2 \implies B_1^2 \leq \frac{2}{3}
            \end{equation}
            giving back the RS solution. 
        
    \end{itemize} 
    We summarise the solutions for $0 < p <1$ in the following table:
    \begin{center}
    \renewcommand{\arraystretch}{1.5}
    \begin{tabular}{c||c||c}
        Case & Condition & Solution of SP equations \\
        \hline
        1a, 1b &
        $-A_2/B_1^2 \leq p / 2$ & 
        No solution 
        \\
        1c &
        $p/2 < -A_2/B_1^2 \leq p$  &
        $c^2 = 1$ and $\mu = \frac{1}{bp} + \frac{bp}{2}$ valid for $b p \geq \sqrt{2}  $
        \\
        1d &
        $p < -A_2/B_1^2 < 1$  &
        $c^2 = \frac{K(p, 2, - \sqrt{2} bc, bc)}{K(p, 1, - \sqrt{2} bc, bc)}$ and 
        $\mu = \frac{b}{2} (1-c^2) - \sqrt{2} c$ valid for $ bcp < \sqrt{2}  < bc $
        \\
        1e, 1f&
        $-A_2/B_1^2 \geq 1$ &
        RS solution $c^2 = 0$ and $\mu = \frac{1}{b} + \frac{b}{2}$
        \\
    \end{tabular}
    \renewcommand{\arraystretch}{1.}
    \end{center}

    \item \textbf{Case $1\leq p < 2$.} We have the subcases:
    
    \begin{center}
    \renewcommand{\arraystretch}{1.5}
    \begin{tabular}{c||c||c|c|c}
        Case & Condition & $\ell=0$ & $\ell=1$ & $\ell=2$ \\
        \hline
        2a & 
        $-A_2/B_1^2 < p/2$ &
        $e^{f(p) }G\left( \frac{A}{B} + p B \right)$ &
        $e^{f(p) }G\left( \frac{A}{B} + p B \right)$ &
        $e^{f(p) }G\left( \frac{A}{B} + p B \right)$
        \\
        2b & 
        $-A_2/B_1^2 = p/2$ &
        $2 e^{f(p) }G\left( \frac{p}{2} B \right)$ &
        $e^{f(p) }G\left( \frac{A}{B} + p B \right)$ &
        $e^{f(p) }G\left( \frac{A}{B} + p B \right)$
        \\
        2c & 
        $p/2 < -A_2/B_1^2 < (p+1)/2$ &
        $G\left( - \frac{A}{B} \right)$ &
        $e^{f(p) }G\left( \frac{A}{B} + p B \right)$ &
        $e^{f(p) }G\left( \frac{A}{B} + p B \right)$
        \\
        2d & 
        $-A_2/B_1^2 = (p+1)/2$ &
        $G\left( - \frac{A}{B} \right)$ &
        $2 e^{f(p) }G\left( \frac{p-1}{2} B \right)$ &
        $e^{f(p) }G\left( \frac{A}{B} + p B \right)$
        \\
        2e & 
        $(p+1)/2 < -A_2/B_1^2  \leq p$ &
        $G\left( - \frac{A}{B} \right)$ &
        $e^{f(1)}G\left( - \frac{A}{B} - B \right)$ &
        $e^{f(p) }G\left( \frac{A}{B} + p B \right)$
        \\
        2f & 
        $p < -A_2/B_1^2 < 2$ &
        $G\left( - \frac{A}{B} \right)$ &
        $e^{f(1)}G\left( - \frac{A}{B} - B \right)$ &
        $\frac{\exp\left(- \frac{A^2}{2 B^2}\right) }{\sqrt{2 \pi} \, N }K(p, 2, A_2, B_1)$
        \\
        2g & 
        $2 \leq -A_2/B_1^2$ &
        $G\left( - \frac{A}{B} \right)$ &
        $e^{f(1)}G\left( - \frac{A}{B} - B \right)$ &
        $e^{f(2)}G\left( - \frac{A}{B} - 2 B \right)$
        \\
    \end{tabular}
    \renewcommand{\arraystretch}{1.}
    \end{center}

    leading to the following equations (notice that if $- A_2 > 0$ than $G(-A/B)=1$):
    \begin{itemize}
        \item[2a)] Same as case 1a).
        \item[2b)] Same as case 1a).
        \item[2c)] Same as case 1c).
        \item[2d)] The equations are 
             \begin{equation}
                 2 e^{f(p) }G\left( \frac{p-1}{2} B \right) = m 
                 \mathand 
                 e^{f(p) }G\left( \frac{A}{B} + p B \right) = c^2 m \, .
             \end{equation}
             We start noticing that 
             \begin{equation}
                \frac{A_2}{B_1^2} + p = \frac{p-1}{2}                
             \end{equation}
             so that the two $G$ factors are equal, and the solution is the same as case 1c).            
        \item[2e)] The equations are 
            \begin{equation}
                e^{f(1)}G\left( - \frac{A}{B} - B \right) = m \mathand e^{f(p) }G\left( \frac{A}{B} + p B \right) = c^2m \, .
            \end{equation}
            The first can be solved following case 1e), giving
            \begin{equation}
                A_2 + \frac{1}{2}B_1^2 = -1 
                \implies 
                \mu = \frac{1}{b} + \frac{b}{2} \, .
            \end{equation}
            The second equation implies $c \to 0$ as $f(1) > f(p)$ due to the condition $(p+1)/2 < -A_2 / B_1^2 \leq p$ 
            implying that the l.h.s. is not of order $m$ and giving back the RS solution. 
        \item[2f)] The equations are 
            \begin{equation}
                e^{f(1)}G\left( - \frac{A}{B} - B \right) = m \mathand \frac{\exp\left(- \frac{A^2}{2 B^2}\right) }{\sqrt{2 \pi} \, N }K(p, 2, A_2, B_1) = c^2m \, .
            \end{equation}
            The first can be solved following case 2e). The second equation again implies $c^2 \to 0$ giving the RS solution.
        \item[2g)]
            The equations are 
                \begin{equation}
                    e^{f(1)}G\left( - \frac{A}{B} - B \right) = m \mathand e^{f(2)}G\left( - \frac{A}{B} - 2 B \right) = c^2m \, .
                \end{equation}
                The first can be solved following case 2e). The second equation again implies $c^2 \to 0$ giving the RS solution.
    \end{itemize} 
    We summarise the solutions for $1 \leq p <2$ in the following table:
    \begin{center}
    \renewcommand{\arraystretch}{1.5}
    \begin{tabular}{c||c||c}
        Case & Condition & Solution of SP equations \\
        \hline
        2a, 2b &
        $-A_2/B_1^2 \leq p / 2$ & 
        No solution 
        \\
        2c, 2d & 
        $p/2 < -A_2/B_1^2 \leq (p+1)/2$  & 
        $c^2 = 1$ and $\mu = \frac{1}{bp} + \frac{bp}{2}$ valid for $b^2 p \geq 2$
        \\
        2e, 2f, 2g& 
        $(p+1)/2 < -A_2/B_1^2 $ &
        RS solution $c^2 = 0$ and $\mu = \frac{1}{b} + \frac{b}{2}$
    \end{tabular}
    \renewcommand{\arraystretch}{1.}
    \end{center}

    \item \textbf{Case $p \geq 2$.} We have the subcases:
    
    \begin{center}
    \renewcommand{\arraystretch}{1.5}
    \begin{tabular}{c||c||c|c|c}
        Case & Condition & $\ell=0$ & $\ell=1$ & $\ell=2$ \\
        \hline
        3a & 
        $-A_2/B_1^2 < p/2$ &
        $e^{f(p) }G\left( \frac{A}{B} + p B \right)$ &
        $e^{f(p) }G\left( \frac{A}{B} + p B \right)$ &
        $e^{f(p) }G\left( \frac{A}{B} + p B \right)$
        \\
        3b & 
        $-A_2/B_1^2 = p/2$ &
        $2 e^{f(p) }G\left( \frac{p}{2} B \right)$ &
        $e^{f(p) }G\left( \frac{A}{B} + p B \right)$ &
        $e^{f(p) }G\left( \frac{A}{B} + p B \right)$
        \\
        3c & 
        $p/2 < -A_2/B_1^2 < (p+1)/2$ &
        $G\left( - \frac{A}{B} \right)$ &
        $e^{f(p) }G\left( \frac{A}{B} + p B \right)$ &
        $e^{f(p) }G\left( \frac{A}{B} + p B \right)$
        \\
        3d & 
        $-A_2/B_1^2 = (p+1)/2$ &
        $G\left( - \frac{A}{B} \right)$ &
        $2 e^{f(p) }G\left( \frac{p-1}{2} B \right)$ &
        $e^{f(p) }G\left( \frac{A}{B} + p B \right)$
        \\
        3e & 
        $(p+1)/2 < -A_2/B_1^2 < (p+2)/2$ &
        $G\left( - \frac{A}{B} \right)$ &
        $e^{f(1)}G\left( - \frac{A}{B} - B \right)$ &
        $e^{f(p) }G\left( \frac{A}{B} + p B \right)$
        \\
        3f & 
        $ -A_2/B_1^2 = (p+2)/2$ &
        $G\left( - \frac{A}{B} \right)$ &
        $e^{f(1)}G\left( - \frac{A}{B} - B \right)$ &
        $2 e^{f(p) }G\left( \frac{p-2}{2} B \right)$
        \\
        3g & 
        $ (p+2)/2 \leq -A_2/B_1^2  $ &
        $G\left( - \frac{A}{B} \right)$ &
        $e^{f(1)}G\left( - \frac{A}{B} - B \right)$ &
        $e^{f(2)}G\left( - \frac{A}{B} - 2 B \right)$
        \\
    \end{tabular}
    \renewcommand{\arraystretch}{1.}
    \end{center}

     leading to the following equations (notice that if $- A_2 > 0$ than $G(-A/B)=1$):
    \begin{itemize}
        \item[3a)] Same as case 1a).
        \item[3b)] Same as case 1a).
        \item[3c)] Same as case 2c).
        \item[3d)] Same as case 2d).
        \item[3e)] Same as case 2e).
        \item[3f)] Same as case 2e).
        \item[3g)] Same as case 3g).
    \end{itemize} 
        We summarise the solutions for $p \geq 2$ in the following table:
    \begin{center}
    \renewcommand{\arraystretch}{1.5}
    \begin{tabular}{c||c||c}
        Case & Condition & Solution of SP equations \\
        \hline
        3a, 3b &
        $-A_2/B_1^2 \leq p / 2$ & 
        No solution 
        \\
        3c, 3d & 
        $p/2 < -A_2/B_1^2 \leq (p+1)/2$  & 
        $c^2 = 1$ and $\mu = \frac{1}{bp} + \frac{bp}{2}$ valid for $b^2 p \geq 2$
        \\
        3e, 3f, 3g& 
        $(p+1)/2 < -A_2/B_1^2 $ &
        RS solution $c^2 = 0$ and $\mu = \frac{1}{b} + \frac{b}{2}$
    \end{tabular}
    \renewcommand{\arraystretch}{1.}
    \end{center}
\end{enumerate}

To sum up all solution found, we have the RS solution, plus
\begin{itemize}
    \item $c^2 = 1$ and $\mu = \frac{1}{bp} + \frac{bp}{2}$ valid for $b \geq \max(\sqrt{2} / p, \sqrt{2/p})$.
    \item 
    $c^2 = \frac{K(p, 2, - \sqrt{2} bc, bc)}{K(p, 1, - \sqrt{2} bc, bc)}$ and 
        $\mu = \frac{b}{2} (1-c^2) - \sqrt{2} c$ valid for $ bcp < \sqrt{2}  < bc $ and $0<p<1$.
\end{itemize}

\subsubsection{Stability}
We can immediately discuss the stability of the 1-RSB solutions found.
The condition is 
\begin{equation}
    \begin{split}
        \beta^2 \int Du \, \frac{\int Dv \, (1+e^{\beta H})^p \left[ \ell(u,v)^4 - 2 \ell(u,v)^3 + \ell(u,v)^2 \right] }{\int Dv \, (1+e^{\beta H})^p} < 1,
    \end{split}
\end{equation}
with $\ell(u,v) = \logistic\left(\beta H(u,v)\right)$, and in the scaling limit it reads
\begin{equation}
    \begin{split}
        b^2 \frac{1}{m} \log \frac{1}{m} \int Dv \, \left(1 + e^{\beta H}\right)^p  \left[
        \ell(\beta H)^4
        - 2 \ell(\beta H)^3
        + \ell(\beta H)^2
        \right] < 1
    \end{split}
\end{equation}
where all denominators and subleading terms in the scaling of $\beta$ were discarded at leading order.
\begin{itemize}
    \item In the RS solution, all integrals of powers of $\ell(u, v)$ are subleading, i.e. they are much smaller than $m$. 
    Thus, their scaling produces an overall positive power of $m$, so that the l.h.s. vanishes as $m \to 0$. 
    Thus the RS solution is stable.
    \item In the $c^2 = 1$ solution, all integrals of power of $\ell(u, v)$ are the same at leading order, so that the same argument as for the RS solution applies. This solution is stable for all $p$.
    \item In the $c^2 = \frac{K(p, 2, - \sqrt{2} bc, bc)}{K(p, 1, - \sqrt{2} bc, bc)}$ solution, the stability condition reads 
        \begin{equation}
            b^2 \log \frac{1}{m} \frac{K(p, 4, - \sqrt{2} bc, bc) -2 K(p, 3, - \sqrt{2} bc, bc) + K(p, 2, - \sqrt{2} bc, bc)}{K(p, 1, - \sqrt{2} bc, bc)} < 1 \, .
        \end{equation}
    Notice also that by construction the l.h.s. of the condition is positive, as the integrand is a positive function times a perfect square.
    Thus, this solution is stable if and only if 
    \begin{equation}
        K(p, 4, - \sqrt{2} bc, bc) -2 K(p, 3, - \sqrt{2} bc, bc) + K(p, 2, - \sqrt{2} bc, bc) = 0
    \end{equation}
    This equation should be thought of as a condition linking $p$ and $b$. This solution may be stable only on a line in the $(p, b)$ plane, and for this reason we discard it.
\end{itemize}

\subsubsection{Complexity}
We need to compute the complexity $\Sigma(p, b)$ only for the stable solution $c^2 = 1$, $\mu = \frac{1}{bp} + \frac{bp}{2}$, valid for $b \geq \max(\sqrt{2} / p, \sqrt{2/p})$.
Recall that 
\begin{equation}
    \begin{split}
        \Sigma(p, b) &= 
        \frac{1}{4} p^2 b^2 c^4 m \log \frac{1}{m} 
        -
        p \int Dv \, (1+e^{\beta H})^p \logpexp(\beta H)
        +
        \log \int Dv \, (1+e^{\beta H})^p
    \end{split}
\end{equation}
so we need to compute the leading asymptotic scaling of 
\begin{equation}
    \int Dv \, (1+e^{\beta H})^p \logpexp(\beta H) 
    \mathand
    \int Dv \, (1+e^{\beta H})^p \, .
\end{equation}
It is immediate to see that the second term is subleading w.r.t. the non-integral term of the complexity.
The first term can be treated using Proposition \ref{prop.log1pexp}, as well as the notations defined in the previous section.
Recall that in this solution $f(p) = \log m$. Then, for $0\leq p <1$
\begin{equation}
    \begin{split}
        \int Dv \, (1+e^{\beta H})^p \logpexp(\beta H) 
        \approx
        \left[A_2 + p B_1^2 \right] m \log \frac{1}{m} G\left( \frac{A}{B} + p B \right)
    \end{split}
\end{equation}
where we used that if $- \frac{A_2}{B_1^2} > 1 > p$ then $f(1) < f(p)$, that $f(p) \geq -\frac{A^2}{2 B^2}$ always holds and that in this solution $- \frac{A_2}{B_1^2} \leq p$.
For $p \geq 1$ instead
\begin{equation}
    \begin{split}
        \int Dv \, (1+e^{\beta H})^p \logpexp(\beta H) 
        \approx
        \left[A_2 + p B_1^2 \right] m \log \frac{1}{m}      \end{split}
\end{equation}
where we used that in this solution and for $p > 1$ then $p/2 < - \frac{A_2}{B_1^2} \leq (p+1)/2 < p $, implying also that $f(1) < f(p)$.
Thus, for $bp > \sqrt{2}$
\begin{equation}
    \frac{\Sigma(p, b)}{\log \frac{1}{m}} = 
        \frac{1}{4} p^2 b^2 m 
        -
        p \left[-b \mu + p b^2 \right]     =     1 - \frac{(bp)^2}{4}  
\end{equation}
while for $bp = \sqrt{2}$ a non-trivial contribution from a $G$ function must be taken into account.

\subsection{Derivation of the asymptotic behaviour of relevant integrals}

\begin{proposition}\label{prop.logistic}
    Consider the integral ($p, \ell \in \bbR$)
    \begin{equation}
        I_{p, \ell}(A, B) 
        = \int Dv \, (1 + e^{A_N + B_N v})^{p-\ell} e^{\ell(A_N + B_N v)} \, ,
    \end{equation}
    where $Dv$ is the standard Gaussian measure.
        Define
    \begin{equation}
        A_2 = \lim_{N \to \infty} \frac{A_N}{N^2} 
        \mathand
        B_1 = \lim_{N \to \infty} \frac{B_N}{N} 
    \end{equation}
    and suppose that $B_1 > 0$ (if $B_1 < 0$, change variable $v \to -v$ in the integral) and $|A_2| < +\infty$.
    Then, as $N \to \infty$, the leading order of $I_{p, \ell}(A, B)$ satisfies
    \begin{equation}
        \begin{split}
            I_{p, \ell}(A, B) \approx 
            \begin{cases}
                \frac{\exp\left(- \frac{A^2}{2 B^2}\right) }{\sqrt{2 \pi} \, N }K(p, \ell, A_2, B_1)
                & p < -\frac{A_2}{B_1^2} < \ell
                \\
                e^{f(p) }G\left( \frac{A}{B} + p B \right) 
                & -\frac{A_2}{B_1^2} \leq p
                \\
                e^{f(\ell)}G\left( - \frac{A}{B} - \ell B \right)
                & -\frac{A_2}{B_1^2} \geq \ell
            \end{cases}
        \end{split}
    \end{equation}
    for $p < \ell$,  and
    \begin{equation}
        \begin{split}
            I_{p, \ell}(A, B) \approx 
            \begin{cases}
                e^{f(p) }G\left( \frac{A}{B} + p B \right) 
                & -\frac{A_2}{B_1^2} < \frac{\ell + p}{2} 
                \\
                2 e^{f(\ell)}G\left( \frac{p - \ell}{2} B\right)
                & -\frac{A_2}{B_1^2} = \frac{\ell + p}{2} 
                \\
                e^{f(\ell)}G\left( - \frac{A}{B} - \ell B \right)
                & -\frac{A_2}{B_1^2} > \frac{\ell + p}{2} 
            \end{cases}
        \end{split}
    \end{equation}
    for $p \geq \ell$, 
    where
    \begin{equation}
        \begin{split}
            f(k) = kA + \frac{1}{2}B^2 k^2 \, ,
        \end{split}
    \end{equation}
    \begin{equation}
        K(p, \ell, a, b) = \frac{b \, _2F_1\left(\frac{a}{b^2}+\ell,\ell-p;\frac{a}{b^2}+\ell+1;-1\right)}{a+b^2 \ell}-\frac{b \, _2F_1\left(-\frac{a}{b^2}-p,\ell-p;-\frac{a}{b^2}-p+1;-1\right)}{a+b^2 p} \, , 
    \end{equation}
    \begin{equation}
        G(z) = \frac{1}{2} \left(1 + \erf \left( \frac{z}{\sqrt{2}} \right)\right) \, ,
    \end{equation}
    and $G(+\infty) = 1$.
\end{proposition}
\begin{proof}
    The first step is to expand the one-plus-exp term. 
    Define $C(v) = A + Bv$ and notice that it is an increasing function of $v$, with unique zero in $v_* = -A/B = \caO_N(N)$.
    We have two uniformly convergent expansions (valid for all $p, \ell \in \bbR$) for $v < v_*$ and $v > v_*$ respectively:
\begin{equation}
    \begin{split}
        (1+e^{C(v)})^{p-\ell} e^{\ell C(v)}
        &= \sum_{n \geq 0} \binom{p-\ell}{n} e^{(n+\ell) C(v)} \mathfor v < v_* \\
        (1+e^{C(v)})^{p-\ell}  e^{\ell C(v)}
        &= (1+e^{-C(v)})^{p-\ell}  e^{pC(v)}
        = \sum_{n \geq 0} \binom{p-\ell}{n} e^{(p - n)C(v)} \mathfor v > v_* \\
    \end{split}
\end{equation}
The two expansions can be integrated in the respective domains term-by-term using the following integrals
\begin{equation}
    \begin{split}
        \int_{-\infty}^{x_1} Dv \, e^{x_2v} &= e^{x_2^2/2} G\left( x_1-x_2 \right)
        \\
        \int_{x_1}^{+\infty} Dv \, e^{x_2v} &= e^{x_2^2/2} G\left( x_2-x_1 \right)
        \\
        G(z) &= \frac{1}{2} \left(1 + \erf \left( \frac{z}{\sqrt{2}} \right)\right)
    \end{split}
\end{equation}
so that
\begin{equation}
    \begin{split}
        I_{p, \ell}(A, B) 
        &=
        \sum_{n \geq 0} \binom{p-\ell}{n} 
        e^{(n+\ell) A + \frac{1}{2}(n+\ell)^2 B^2}
        G\left( - \frac{A}{B} - (n+\ell) B \right)
        \\
        &\quad+
        \sum_{n \geq 0} \binom{p-\ell}{n} e^{(p - n) A+ \frac{1}{2}(p-n)^2 B^2}
        G\left( (p-n)B + \frac{A}{B} \right) \, ,
    \end{split}
\end{equation}
where, in both cases, the order of the argument of the $G$ function is a most $\caO_N(N)$.
Then, we use the following asymptotic expansion for the $G$ functions
\begin{equation}
    G(z_1 N + z_0 + z_{-1} / N + \dots) \sim
    \begin{cases}
        1 & z_1 > 0 \\ 
        G(z_1) & z_1 = 0 \\
        \frac{1}{\sqrt{2 \pi} \, |z_1| N }
        \exp\left( -\frac{1}{2}
        \left(
            z_1^2 N^2
            +2 z_1 z_0 N
            + z_0^{-1}
            +2 z_1 z_{-1}
            + \dots
        \right) \right) & z_1 < 0  
    \end{cases} \, ,
\end{equation}
where again one needs to be careful in considering all non-vanishing subleading orders in the exponential.
Notice that with the abuse of notation $G(+\infty) = 1$ we have 
\begin{equation}
    G(z = z_1 N + z_0 + z_{-1} / N + \dots) \sim
    \begin{cases}
        G(z) & z_1 \geq 0 \\ 
        \left(\sqrt{2 \pi} \, |z_1| N \right)^{-1}
        \exp\left( -z^2 / 2\right) & z_1 < 0  
    \end{cases} \, ,
\end{equation}
So we have that
\begin{equation}
    \begin{split}
        e^{k A + \frac{1}{2}k^2 B^2}
        G\left( - \frac{A}{B} - k B \right)
        &\sim
        \begin{cases}
            e^{k A + \frac{1}{2}k^2 B^2}G\left( - \frac{A}{B} - k B \right) 
            & - \left( k B_1 + \frac{A_2}{B_1} \right)  \geq 0 \\
            \left(\sqrt{2 \pi} \, |k B_1 + \frac{A_2}{B_1}| N \right)^{-1}
            \exp\left(
                - \frac{1}{2} \frac{A^2}{B^2}
            \right)  & - \left( k B_1 + \frac{A_2}{B_1} \right)  < 0  
        \end{cases} \, ,
    \end{split}
\end{equation}
and 
\begin{equation}
    \begin{split}
        e^{k A + \frac{1}{2}k^2 B^2}
        G\left( \frac{A}{B} + k B \right)
        &\sim
        \begin{cases}
            e^{k A + \frac{1}{2}k^2 B^2}G\left( \frac{A}{B} + k B \right) 
            & \left( k B_1 + \frac{A_2}{B_1} \right)  \geq 0 \\
            \left(\sqrt{2 \pi} \, |k B_1 + \frac{A_2}{B_1}| N \right)^{-1}
            \exp\left(
                - \frac{1}{2} \frac{A^2}{B^2}
            \right)  & \left( k B_1 + \frac{A_2}{B_1} \right)  < 0  
        \end{cases} \, ,
    \end{split}
\end{equation}
so that
\begin{equation}
    \begin{split}
        I_{p, \ell}(A, B) 
        &\approx
        \sum_{n \geq 0} \binom{p-\ell}{n} 
        \begin{cases}
            e^{(n+\ell) A + \frac{1}{2}(n+\ell)^2 B^2}G\left( - \frac{A}{B} - (n+\ell) B \right) 
            & \left( (n+\ell) B_1 + \frac{A_2}{B_1} \right)  \leq 0 \\
            \left(\sqrt{2 \pi} \, |(n+\ell) B_1 + \frac{A_2}{B_1}| N \right)^{-1}
            \exp\left(
                - \frac{1}{2} \frac{A^2}{B^2}
            \right)  & \left( (n+\ell) B_1 + \frac{A_2}{B_1} \right)  > 0  
        \end{cases}
        \\
        &\quad+
        \sum_{n \geq 0} \binom{p-\ell}{n} 
        \begin{cases}
            e^{(p-n) A + \frac{1}{2}(p-n)^2 B^2}G\left( \frac{A}{B} + (p-n) B \right) 
            & \left( (p-n) B_1 + \frac{A_2}{B_1} \right)  \geq 0 \\
            \left(\sqrt{2 \pi} \, |(p-n) B_1 + \frac{A_2}{B_1}| N \right)^{-1}
            \exp\left(
                - \frac{1}{2} \frac{A^2}{B^2}
            \right)  & \left( (p-n) B_1 + \frac{A_2}{B_1} \right)  < 0  
        \end{cases} 
        \\
        &=
        \sum_{n \geq 0} \binom{p-\ell}{n} 
        \begin{cases}
            e^{(n+\ell) A + \frac{1}{2}(n+\ell)^2 B^2}G\left( - \frac{A}{B} - (n+\ell) B \right) 
            & n \leq -\frac{A_2}{B_1^2} - \ell \\
            \left(\sqrt{2 \pi} \, |(n+\ell) B_1 + \frac{A_2}{B_1}| N \right)^{-1}
            \exp\left(
                - \frac{1}{2} \frac{A^2}{B^2}
            \right)  & n > -\frac{A_2}{B_1^2} - \ell
        \end{cases}
        \\
        &\quad+
        \sum_{n \geq 0} \binom{p-\ell}{n} 
        \begin{cases}
            e^{(p-n) A + \frac{1}{2}(p-n)^2 B^2}G\left( \frac{A}{B} + (p-n) B \right) 
            & n \leq \frac{A_2}{B_1^2} + p \\
            \left(\sqrt{2 \pi} \, |(p-n) B_1 + \frac{A_2}{B_1}| N \right)^{-1}
            \exp\left(
                - \frac{1}{2} \frac{A^2}{B^2}
            \right)  & n > \frac{A_2}{B_1^2} + p
        \end{cases} 
        \\
        \\
        &=
        \sum_{n \geq 0} \binom{p-\ell}{n} 
        \begin{cases}
            e^{f(n+\ell)}G\left( - \frac{A}{B} - (n+\ell) B \right) 
            +
            e^{f(p-n)}G\left( \frac{A}{B} + (p-n) B \right)
            & n \leq \min \left( -\frac{A_2}{B_1^2} - \ell , \frac{A_2}{B_1^2} + p \right)\\
            e^{f(n+\ell) }G\left( - \frac{A}{B} - (n+\ell) B \right) 
            +
            \frac{\exp\left(
                f(-A/B^2)
            \right) }
            {
                \sqrt{2 \pi} \, |(p-n) B_1 + \frac{A_2}{B_1}| N 
            }
            &  \frac{A_2}{B_1^2} + p <   n \leq      -\frac{A_2}{B_1^2} - \ell             \\
            \frac{\exp\left(
                f(-A/B^2)
            \right) }
            {\sqrt{2 \pi} \, |(n+\ell) B_1 + \frac{A_2}{B_1}| N } 
            +
            e^{f(p-n) }G\left( \frac{A}{B} + (p-n) B \right) 
            &   -\frac{A_2}{B_1^2} - \ell <   n \leq        \frac{A_2}{B_1^2} + p          \\
            \frac{\exp\left(
                f(-A/B^2)
            \right) }
            {\sqrt{2 \pi} \, |(n+\ell) B_1 + \frac{A_2}{B_1}| N } 
            +
            \frac{\exp\left(
                f(-A/B^2)
            \right) }
            {
                \sqrt{2 \pi} \, |(p-n) B_1 + \frac{A_2}{B_1}| N 
            }
            & n > \max \left( -\frac{A_2}{B_1^2} - \ell , \frac{A_2}{B_1^2} + p \right)
        \end{cases}
    \end{split}
\end{equation}
where we defined
\begin{equation}
    \begin{split}
        f(k) = kA + \frac{1}{2}B^2 k^2 = \left( kA_2 + \frac{1}{2}B_1^2 k^2 \right) N^2 + \dots = f_{2}(k) N^2 + f_1(k) N + f_0(k) + \dots
    \end{split}
\end{equation}
whose leading order is $N^2$ times a parabola $f_2(k)$ with zeros in $k=0, -2A_2/B_1^2$ and vertex in $k = -A_2/B_1^2$, and positive concavity (recall that $B_1 > 0$).

We see then that $I_{p, \ell}(A, B)$ can be written as a sum of exponentials, whose exponents are of order $\caO_N(N^2)$.
The leading term of the integral can be obtained by finding the terms whose $N^2$ coefficient in the exponent is the largest.
For this reason we can immediately say that, as  $f(-A/B^2) \leq f(k)$ for all $k \in \bbR$ and given the additional factor $1/N$ in the terms with exponent $f(-A/B^2)$, 
\begin{equation}
    \begin{split}
        I_{p, \ell}(A, B) 
        &\approx
        \sum_{n \geq 0} \binom{p-\ell}{n} 
        \begin{cases}
            e^{f(n+\ell)}G\left( - \frac{A}{B} - (n+\ell) B \right) 
            +
            e^{f(p-n)}G\left( \frac{A}{B} + (p-n) B \right)
            & n \leq \min \left( -\frac{A_2}{B_1^2} - \ell , \frac{A_2}{B_1^2} + p \right)\\
            e^{f(n+\ell) }G\left( - \frac{A}{B} - (n+\ell) B \right) 
            &  \frac{A_2}{B_1^2} + p <   n \leq      -\frac{A_2}{B_1^2} - \ell             \\
            e^{f(p-n) }G\left( \frac{A}{B} + (p-n) B \right) 
            &   -\frac{A_2}{B_1^2} - \ell <   n \leq        \frac{A_2}{B_1^2} + p          \\
            \frac{\exp\left(
                f(-A/B^2)
            \right) }
            {\sqrt{2 \pi} \, |(n+\ell) B_1 + \frac{A_2}{B_1}| N } 
            +
            \frac{\exp\left(
                f(-A/B^2)
            \right) }
            {
                \sqrt{2 \pi} \, |(p-n) B_1 + \frac{A_2}{B_1}| N 
            }
            & n > \max \left( -\frac{A_2}{B_1^2} - \ell , \frac{A_2}{B_1^2} + p \right)
        \end{cases}
    \end{split}
\end{equation}
Now
\begin{itemize}
    \item if 
        \begin{equation}
            \max \left( -\frac{A_2}{B_1^2} - \ell , \frac{A_2}{B_1^2} + p \right) < 0
            \iff p < -\frac{A_2}{B_1^2} < \ell
        \end{equation}
        only the fourth branch contributes (notice that this can happen only if $p < \ell$). Each term in the sum is of the same order in $N$, giving 
        \begin{equation}
            \begin{split}
                I_{p, \ell}(A, B) &\approx 
            \frac{\exp\left(- \frac{A^2}{2 B^2}\right) }{\sqrt{2 \pi} \, N }
            \sum_{n \geq 0} \binom{p-\ell}{n} 
            \left[\frac{1}{-(p-n) B_1 - \frac{A_2}{B_1}} + \frac{1}{(n+\ell) B_1 + \frac{A_2}{B_1}}\right]
            \\
            &=
            \frac{\exp\left(- \frac{A^2}{2 B^2}\right) }{\sqrt{2 \pi} \, N }K(p, \ell, A_2, B_1)
            \end{split}
        \end{equation}
        where
        \begin{equation}
            K(p, \ell, A_2, B_1) = \frac{B_1 \, _2F_1\left(\frac{A_2}{B_1^2}+\ell,\ell-p;\frac{A_2}{B_1^2}+\ell+1;-1\right)}{A_2+B_1^2 \ell}-\frac{B_1 \, _2F_1\left(-\frac{A_2}{B_1^2}-p,\ell-p;-\frac{A_2}{B_1^2}-p+1;-1\right)}{A_2+B_1^2 p}
        \end{equation}
    \item if instead 
         \begin{equation}
            \max \left( -\frac{A_2}{B_1^2} - \ell , \frac{A_2}{B_1^2} + p \right) \geq 0
        \end{equation}
        the fourth branch never contributes. It is easy to see graphically that in this case, the leading term is $n=0$, whichever branch it is contained into. Indeed, one plots the lines 
        \begin{equation}
            n = \frac{A_2}{B_1^2} + p \mathand n = -\frac{A_2}{B_1^2} - \ell
        \end{equation}
        as a function of $A_2 / B_1^2$. Then each of the four sectors identified by the intersection of the two lines identifies one of the branches (the leftmost identifies the third branch, the bottom one the fourth and the right one the second). Then in all the sectors it is easy to see that the exponents $f(p-n)$ and $f(n +\ell)$ are monotonically decreasing, so that they separately achieve their maxima at $n=0$. 
        Finally, we notice that 
        \begin{equation}
            f(p) > f(\ell) \iff -\frac{A_2}{B_1^2} < \frac{\ell + p}{2} 
                    \end{equation} 
        giving
        \begin{equation}
            \begin{split}
                I_{p, \ell}(A, B) \approx 
                \begin{cases}
                    e^{f(p) }G\left( \frac{A}{B} + p B \right) 
                    & -\frac{A_2}{B_1^2} < \frac{\ell + p}{2} \\
                    e^{f(p) }G\left( \frac{A}{B} + p B \right)
                    +
                    e^{f(\ell)}G\left( - \frac{A}{B} - \ell B \right)
                    & -\frac{A_2}{B_1^2} = \frac{\ell + p}{2} \\
                    e^{f(\ell)}G\left( - \frac{A}{B} - \ell B \right)
                    & -\frac{A_2}{B_1^2} > \frac{\ell + p}{2} 
                \end{cases}
            \end{split}
        \end{equation}
\end{itemize}
Thus, to summarise, for $p < \ell$
\begin{equation}
    \begin{split}
        I_{p, \ell}(A, B) \approx 
        \begin{cases}
            \frac{\exp\left(- \frac{A^2}{2 B^2}\right) }{\sqrt{2 \pi} \, N }K(p, \ell, A_2, B_1)
            & p < -\frac{A_2}{B_1^2} < \ell
            \\
            e^{f(p) }G\left( \frac{A}{B} + p B \right) 
            & -\frac{A_2}{B_1^2} \leq p
            \\
            e^{f(\ell)}G\left( - \frac{A}{B} - \ell B \right)
            & -\frac{A_2}{B_1^2} \geq \ell
        \end{cases}
    \end{split}
\end{equation}
and for $p \geq \ell$
\begin{equation}
    \begin{split}
        I_{p, \ell}(A, B) \approx 
        \begin{cases}
            e^{f(p) }G\left( \frac{A}{B} + p B \right) 
            & -\frac{A_2}{B_1^2} < \frac{\ell + p}{2} 
            \\
            e^{f(p) }G\left( \frac{A}{B} + p B \right)
            +
            e^{f(\ell)}G\left( - \frac{A}{B} - \ell B \right)
            & -\frac{A_2}{B_1^2} = \frac{\ell + p}{2} 
            \\
            e^{f(\ell)}G\left( - \frac{A}{B} - \ell B \right)
            & -\frac{A_2}{B_1^2} > \frac{\ell + p}{2} 
        \end{cases}
    \end{split}
\end{equation}
\end{proof}

\begin{proposition}\label{prop.log1pexp}
Consider the integral ($p, \ell \in \bbR$)
    \begin{equation}
        J_{p}(A, B) 
        = \int Dv \, (1 + e^{A_N + B_N v})^{p} \log\left(1 + e^{\ell(A_N + B_N v)} \right) \, ,
    \end{equation}
    where $Dv$ is the standard Gaussian measure. Under the same hypothesis and notations of  Proposition \ref{prop.logistic}, we have that, as $N \to \infty$, the leading order of $J_{p}(A, B)$ satisfies for $0 \leq p < 1$
    \begin{equation}
    \begin{split}
        J_p(A, B) \approx
            \begin{cases}
                                                                \left[A + p B^2 \right] e^{f(p)}G\left( \frac{A}{B} + p B \right)
                &  - \frac{A_2}{B_1^2} \leq p \\
                \frac{B_1 + K_3(p, A_2, B_1)}{\sqrt{2\pi}} N  \exp\left( -\frac{A^2}{2 B^2} \right)
                & p < - \frac{A_2}{B_1^2} < 1 \\
                e^{f(1)}G\left( - \frac{A}{B} - B \right)
                & - \frac{A_2}{B_1^2} \geq 1 \\
            \end{cases}
        \end{split}
    \end{equation}
    and for $p \geq 1$
    \begin{equation}
        \begin{split}
            J_p(A, B) \approx
            \begin{cases}
                                                \left[A + p B^2 \right] e^{f(p)}G\left( \frac{A}{B} + p B \right)
                & - \frac{A_2}{B_1^2} <1  \\
                e^{f(1)}G\left( - \frac{A}{B} - B \right)
                                + \left[A + p B^2 \right] e^{f(p)}G\left( \frac{A}{B} + p B \right)
                & 1 \leq - \frac{A_2}{B_1^2} \leq p \\
                                                                e^{f(1)}G\left( - \frac{A}{B} - B \right)
                & - \frac{A_2}{B_1^2} > p \\
            \end{cases}
        \end{split}
    \end{equation}
    where 
    \begin{equation}
        \begin{split}
            f(k) = kA + \frac{1}{2}B^2 k^2 \, ,
        \end{split}
    \end{equation}
        \begin{equation}
        K_3(p, a, b)
        =
        -a \frac{b \, _2F_1\left(-\frac{a}{b^2}-p,-p;-\frac{a}{b^2}-p+1;-1\right)}{a+b^2 p}
        - pb^2 \frac{b \, _2F_1\left(1-p,-\frac{a}{b^2}-p+1;-\frac{a}{b^2}-p+2;-1\right)}{a+b^2 (p-1)} \, ,
    \end{equation}
    \begin{equation}
        G(z) = \frac{1}{2} \left(1 + \erf \left( \frac{z}{\sqrt{2}} \right)\right) \, ,
    \end{equation}  
    and $G(+\infty) = 1$.
\end{proposition}
\begin{proof}
Define $C(v) = A + Bv$ and notice that it is an increasing function of $v$, with unique zero in $v_* = -A/B = \caO_N(N)$.
Then notice that 
\begin{align}
\logpexp(x) = \max(x,0) +\log\left(1+e^{-|x|}\right) ,
\end{align}
and for $x\neq0$ we have that $e^{|x|}<1$ so that we can write a uniformly convergent series by expanding the logarithm in Taylor series i.e.
\begin{align}
    \logpexp(x) = \max(x,0) +\sum_{n\geq 1}d_ne^{-n|x|}; \hspace{1em}\text{where}\hspace{1em} d_n = \frac{(-1)^{n+1}}{n}.
\end{align}
Thus, we have the following uniformly convergent expansions for the integrand (for some sets of coefficients $e_n$ and $f_n$ such that  $e_1 = 1$ and $f_1 = 1$), whose sum equals exaclty $J_p(A, B)$
\begin{equation}
    \begin{split}
        I_1 =  \int_{-\infty}^{v^*}&Dv\,(1 + e^{C(v)})^{p}\logpexp(C(v)) 
        = \sum_{n\geq1}e_n \int_{-\infty}^{v^*}Dv\, e^{nC(v)} \\
        &=
        \sum_{n\geq1}e_n
        \begin{cases}
            e^{n A + \frac{1}{2}n^2 B^2}G\left( - \frac{A}{B} - n B \right) 
            & n \leq -\frac{A_2}{B_1^2} \\
            \left(\sqrt{2 \pi} \, |n B_1 + \frac{A_2}{B_1}| N \right)^{-1}
            \exp\left(
                - \frac{1}{2} \frac{A^2}{B^2}
            \right)  & n > -\frac{A_2}{B_1^2} 
        \end{cases}
        \\
         I_2 =  \int^{+\infty}_{v^*}&Dv\,(1 + e^{C(v)})^{p}\logpexp(C(v)) 
         = \sum_{n\geq1}f_n \int^{+\infty}_{v^*}Dv\, e^{(p-n)C(v)} \\
         &= 
         \sum_{n\geq1}f_n
         \begin{cases}
            e^{(p-n) A + \frac{1}{2}(p-n)^2 B^2}G\left( \frac{A}{B} + (p-n) B \right) 
            & n \leq \frac{A_2}{B_1^2} + p \\
            \left(\sqrt{2 \pi} \, |(p-n) B_1 + \frac{A_2}{B_1}| N \right)^{-1}
            \exp\left(
                - \frac{1}{2} \frac{A^2}{B^2}
            \right)  & n > \frac{A_2}{B_1^2} + p
        \end{cases} 
        \\
         I_3 =  \int_{v_*}^{+\infty}&Dv\,(1 + e^{C(v)})^{p} C(v) 
         = A \int_{v_*}^{+\infty}Dv\,(1 + e^{C(v)})^{p} 
         + B \int_{v_*}^{+\infty}Dv\,(1 + e^{C(v)})^{p} v \\
         &= A \int_{v_*}^{+\infty}Dv\,(1 + e^{C(v)})^{p} 
         +  p B^2 \int_{v_*}^{+\infty}Dv\,(1 + e^{C(v)})^{p-1} e^{C(v)}
         +  \frac{B}{\sqrt{2\pi}} e^{ -\frac{A^2}{2 B^2} } \\
         &=
         \frac{B}{\sqrt{2\pi}} e^{ -\frac{A^2}{2 B^2} }
         \\&\quad
         + \sum_{n \geq 0} \left[A \binom{p}{n} + pB^2  \binom{p-1}{n} \right]
        \begin{cases}
            e^{(p-n) A + \frac{1}{2}(p-n)^2 B^2}G\left( \frac{A}{B} + (p-n) B \right) 
            & n \leq \frac{A_2}{B_1^2} + p \\
            \left(\sqrt{2 \pi} \, |(p-n) B_1 + \frac{A_2}{B_1}| N \right)^{-1}
            \exp\left(
                - \frac{1}{2} \frac{A^2}{B^2}
            \right)  & n > \frac{A_2}{B_1^2} + p
        \end{cases} 
    \end{split}
\end{equation}
As done in the previous proposition, we need to study the leading term of each of the three integrals $I_{1,2,3}$. We have that
\begin{equation}
    \begin{split}
        I_1 \approx 
        \begin{cases}
            e^{A + \frac{1}{2}B^2}G\left( - \frac{A}{B} - B \right)
            & -\frac{A_2}{B_1^2} \geq 1 \\
            \caO\left(\exp\left( -\frac{A^2}{2 B^2} \right) / \sqrt{\log 1/m}\right) 
            & -\frac{A_2}{B_1^2} < 1
        \end{cases}
    \end{split}
\end{equation}
as whenever the first branch contributes, by increasing monotonicity of the exponent for $1 \leq n \leq -\frac{A_2}{B_1^2}$, the $n=1$ term dominates.
Similarly
\begin{equation}
    \begin{split}
        I_2 \approx 
        \begin{cases}
            e^{(p-1)A + \frac{1}{2}(p-1)^2B^2}G\left( \frac{A}{B} +(p-1) B \right)
            & - \frac{A_2}{B_1^2}  \leq p-1 \\
            \caO\left(\exp\left( -\frac{A^2}{2 B^2} \right) / \sqrt{\log 1/m}\right) 
            & - \frac{A_2}{B_1^2} > p- 1
        \end{cases}
    \end{split}
\end{equation}
and finally
\begin{equation}
    \begin{split}
        I_3 &\approx 
        \frac{B}{\sqrt{2\pi}} e^{ -\frac{A^2}{2 B^2} } + 
        \begin{cases}
            \left[A + p B^2 \right] e^{pA + \frac{1}{2}p^2B^2}G\left( \frac{A}{B} + p B \right)
            & - \frac{A_2}{B_1^2} \leq p  \\
            \frac{K_3(p, A_2, B_1)}{\sqrt{2\pi}} N  \exp\left( -\frac{A^2}{2 B^2} \right) 
            & - \frac{A_2}{B_1^2} > p
        \end{cases}
        \\
        &\approx 
        \begin{cases}
            \left[A + p B^2 \right] e^{pA + \frac{1}{2}p^2B^2}G\left( \frac{A}{B} + p B \right)
            & - \frac{A_2}{B_1^2} \leq p  \\
                        \frac{B_1 + K_3(p, A_2, B_1)}{\sqrt{2\pi}} N  \exp\left( -\frac{A^2}{2 B^2} \right) 
            & - \frac{A_2}{B_1^2} > p
        \end{cases}
    \end{split}
\end{equation}
where
\begin{equation}
    K_3(p, a, b)
    =
    -a \frac{b \, _2F_1\left(-\frac{a}{b^2}-p,-p;-\frac{a}{b^2}-p+1;-1\right)}{a+b^2 p}
    - pb^2 \frac{b \, _2F_1\left(1-p,-\frac{a}{b^2}-p+1;-\frac{a}{b^2}-p+2;-1\right)}{a+b^2 (p-1)}
\end{equation}

Summing it up, this gives for $0 \leq p<1$

\begin{equation}
    \begin{split}
        J_p(A, B) = I_1 + I_2 + I_3 \approx
        \begin{cases}
                                                \left[A + p B^2 \right] e^{f(p)}G\left( \frac{A}{B} + p B \right)
            &  - \frac{A_2}{B_1^2} \leq p \\
            \frac{B_1 + K_3(p, A_2, B_1)}{\sqrt{2\pi}} N  e^{ -\frac{A^2}{2 B^2} }
            & p < - \frac{A_2}{B_1^2} < 1 \\
            e^{f(1)}G\left( - \frac{A}{B} - B \right)
            & - \frac{A_2}{B_1^2} \geq 1 \\
        \end{cases}
    \end{split}
\end{equation}
where we used in the first branch that the exponent $f(p-1) < f(p)$, 
and for $p \geq 1$
\begin{equation}
    \begin{split}
        J_p(A, B) \approx
        \begin{cases}
                                    \left[A + p B^2 \right] e^{f(p)}G\left( \frac{A}{B} + p B \right)
            & - \frac{A_2}{B_1^2} <1  \\
            e^{f(1)}G\left( - \frac{A}{B} - B \right)
                        + \left[A + p B^2 \right] e^{f(p)}G\left( \frac{A}{B} + p B \right)
            & 1 \leq - \frac{A_2}{B_1^2} \leq p \\
                                                e^{f(1)}G\left( - \frac{A}{B} - B \right)
            & - \frac{A_2}{B_1^2} > p \\
        \end{cases}
    \end{split}
\end{equation}
\end{proof}

\section{Details of the Type-I stability computation} \label{label.stabdetails}

In this Section we derive a stability condition for the 1-RSB anstaz.
The 1-RSB ansatz is stable if the solutions of the associated variational problem (i.e. local extremisers of the 1-RSB variational free entropy \eqref{eq.phiRSB1}) are also extremisers of the original replicated saddle-point problem.
A less stringent and less computationally heavy check is to check that the solutions of the 1-RSB variational problem are stable against 2-RSB perturbations. This gives rise to the so-called Type-I and Type-II stability conditions (each type being associated with a different type of 2-RSB perturbation).

In this section, we provide the Type-I stability condition (i.e. under the 2-RSB perturbation \eqref{eq.typeIperturbation}). We conjecture that this coincides with the Type-II, and the more general stability thresholds, as happens in the usual SK model.

We will use the following notations for the 2-RSB (and analgous for the 1-RSB) saddle-point equations
\begin{equation}
    \begin{split}
        \angavg{f(u)}_u &= \int Du \, f(u) \, ,\\
        \angavg{f(u, v)}_v &= \frac{1}{N_v(u)}\int Dv \, f(u, v) N_z(u, v)^{p_1/p_2} \, ,\\
        \angavg{f(u, v, z)}_z &= \frac{1}{N_z(u, v)}\int Dz \, f(u, v, z) \left(1 + e^{\beta H(u, v, z)}\right)^{p_2} \, ,\\
        N_v(u) &= \int Dv \, N_z(u, v)^{p_1/p_2} \, ,\\
        N_z(u, v) &= \int Dz \, \left(1 + e^{\beta H(u, v, z)}\right)^{p_2} \, ,
    \end{split}
\end{equation}
so that the equation for $q_2$ reads:
\begin{equation}
    \begin{split}
        q_2 = \angavg{\angavg{\angavg{ \ell^2 }_z}_v}_u  \, ,
    \end{split}
\end{equation}where $\ell = \logistic(\beta H_{\rm 2RSB})$ and
\begin{equation}
     H_{\rm 2RSB}\equiv H(u, v, z)  =  h + \frac{\beta}{2} (m - q_2) 
    + \sqrt{q_0} u + \sqrt{q_1 - q_0} v + \sqrt{q_2 - q_1} z
\end{equation}

Let 
\begin{equation}\label{eq.typeIperturbation}
    \begin{split}
        m^{\rm 2RSB} &= m \, , \\
        q^{\rm 2RSB}_2 &= q_1 + \epsilon \, , \\
        q^{\rm 2RSB}_1 &= q_1 \, , \\
        q^{\rm 2RSB}_0 &= q_0\, , \\
        p^{\rm 2RSB}_2 &= x \, , \\
        p^{\rm 2RSB}_1 &= p_1 \, ,\\
    \end{split}
\end{equation}
with $\epsilon > 0$ small, $x \in (p_1, 1)$ and $m, q_1, q_0, p_1$ are the solutions to the 1-RSB SP equations. 

We start by expanding
\begin{equation}
    \begin{split}
        \beta H_{\rm 2RSB} 
        &= 
        \beta h + \frac{\beta^2}{2} (m^{\rm 2RSB} - q^{\rm 2RSB}_2) 
        + \beta \sqrt{q^{\rm 2RSB}_0} u 
        + \beta \sqrt{q^{\rm 2RSB}_1 - q^{\rm 2RSB}_0} v 
        + \beta \sqrt{q^{\rm 2RSB}_2 - q^{\rm 2RSB}_1} z
        \\&= 
        \beta h + \frac{\beta^2}{2} (m - q_1 - \epsilon)
        + \beta \sqrt{q_0} u 
        + \beta \sqrt{q_1 - q_0} v 
        + \beta \sqrt{\epsilon} z
        \\&= \beta H_{\rm 1-RSB}(u, v) 
        + \beta \sqrt{\epsilon} z
        - \frac{\beta^2}{2} \epsilon 
    \end{split}
\end{equation}
so that ($\ell(u,v) = \ell(\beta H_{\rm 1-RSB}(u,v))$ is the logistic function)
\begin{equation}
    \begin{split}
        \left(1+e^{\beta H_{\rm 2RSB}}\right)^{p^{2RSB}_2} 
        &=
        \left(1+e^{\beta H_{\rm 1-RSB}(u, v) + \beta \sqrt{\epsilon} z
        - \frac{\beta^2}{2} \epsilon }\right)^{x}
        \\&=
        \left(1+e^{\beta H_{\rm 1-RSB}(u, v)}\right)^{x}
        \left(1 + \ell(u,v)\left( e^{\beta \sqrt{\epsilon} z
        - \frac{\beta^2}{2} \epsilon} - 1\right)  \right)^{x}
        \\&=
        \left(1+e^{\beta H_{\rm 1-RSB}(u, v)}\right)^{x}
        \left(1 + \ell(u,v)\left( \beta \sqrt{\epsilon} z
        + (z^2-1) \frac{\beta^2}{2} \epsilon + \caO(\epsilon^{3/2})\right)  \right)^{x}
        \\&=
        \left(1+e^{\beta H_{\rm 1-RSB}(u, v)}\right)^{x}
        \left(1 
        + x \ell(u,v)\left( \beta \sqrt{\epsilon} z
        + (z^2-1) \frac{\beta^2}{2} \epsilon \right)  
        + \frac{x(x-1)}{2} \beta^2 \epsilon z^2 \ell(u,v)^2
        + \caO(\epsilon^{3/2})
        \right)
        \\&= 
        \left(1+e^{\beta H_{\rm 1-RSB}(u, v)}\right)^{x}
        \left(1 
        + \sqrt{\epsilon} z \beta x \ell(u,v)
        + \epsilon \frac{\beta^2}{2} x \ell(u,v) \left( z^2 - 1 + (x - 1) z^2 \ell(u,v) \right)
        + \caO(\epsilon^{3/2})
        \right)
    \end{split}
\end{equation}
implying 
\begin{equation}
    \begin{split}
        N^{\rm 2RSB}_z(u, v) &=
        \int Dz \, \left(1+e^{\beta H_{\rm 1-RSB}(u, v)}\right)^{x}
        \left(1 
        + \sqrt{\epsilon} z \beta x \ell(u,v)
        + \epsilon \frac{\beta^2}{2} x \ell(u,v) \left( z^2 - 1 + (x - 1) z^2 \ell(u,v) \right)
        + \caO(\epsilon^{3/2})
        \right)
        \\&=
        \left(1+e^{\beta H_{\rm 1-RSB}(u, v)}\right)^{x}
        \left(1 
        + \epsilon \frac{\beta^2}{2} x (x - 1) \ell(u,v)^2 
        + \caO(\epsilon^{2})
        \right) \, .
    \end{split}
\end{equation}

Moreover 
\begin{equation}
    \begin{split}
        \ell&\left(
            \beta H_{\rm 1-RSB}(u, v) + \beta \sqrt{\epsilon} z
            - \frac{\beta^2}{2} \epsilon 
        \right)
        \\&=
        \ell(u, v) \left( 1
        + \sqrt{\epsilon} z \beta ( 1- \ell(u, v))
        - \epsilon \frac{\beta^2}{2} (1-\ell(u,v)) \left[1-z^2 + 2 z^2 \ell(u,v) \right]
        + \caO(\epsilon^{3/2}) \right)
    \end{split}
\end{equation}
and 
\begin{equation}
    \begin{split}
        \ell&\left(
            \beta H_{\rm 1-RSB}(u, v) + \beta \sqrt{\epsilon} z
            - \frac{\beta^2}{2} \epsilon 
        \right)^2
        \\&=
        \ell(u, v)^2 \left( 1
        + \sqrt{\epsilon} z \beta ( 1- \ell(u, v))
        - \epsilon \frac{\beta^2}{2} (1-\ell(u,v)) \left[1-z^2 + 2 z^2 \ell(u,v) \right]
        + \caO(\epsilon^{3/2}) \right)^2
        \\&=
        \ell(u, v)^2 \left( 1
        + 2 \sqrt{\epsilon} z \beta ( 1- \ell(u, v))
        + \epsilon z^2 \beta^2 ( 1- \ell(u, v))^2
        - \epsilon \beta^2 (1-\ell(u,v)) \left[1-z^2 + 2 z^2 \ell(u,v) \right]
        + \caO(\epsilon^{3/2}) \right)
        \\&=
        \ell(u, v)^2 \left( 1
        + 2 \sqrt{\epsilon} z \beta ( 1- \ell(u, v))
        - \epsilon \beta^2 (1-\ell(u,v)) \left[1 - 2 z^2 + 3 z^2 \ell(u,v) \right]
        + \caO(\epsilon^{3/2}) \right)
    \end{split}
\end{equation}

Thus, we have that 
\begin{equation}
    \begin{split}
        &\angavg{\ell(u,v,z)^2}^{\rm 2RSB}_z = 
        \left(1+e^{\beta H_{\rm 1-RSB}(u, v)}\right)^{- x}
        \left(1 
        - \epsilon \frac{\beta^2}{2} x (x - 1) \ell(u,v)^2 
        + \caO(\epsilon^{2})
        \right) \times 
        \\&\quad\times 
        \int Dz \, 
        \left(1+e^{\beta H_{\rm 1-RSB}(u, v)}\right)^{x}
        \left(1 
        + \sqrt{\epsilon} z \beta x \ell(u,v)
        + \epsilon \frac{\beta^2}{2} x \ell(u,v) \left( z^2 - 1 + (x - 1) z^2 \ell(u,v) \right)
        + \caO(\epsilon^{3/2})
        \right)\times 
        \\&\quad\times 
        \ell(u, v)^2 \left( 1
        + 2 \sqrt{\epsilon} z \beta ( 1- \ell(u, v))
        - \epsilon \beta^2 (1-\ell(u,v)) \left[1 - 2 z^2 + 3 z^2 \ell(u,v) \right]
        + \caO(\epsilon^{3/2}) \right)
        \\&= 
        \ell(u, v)^2
        \left(1 
        - \epsilon \frac{\beta^2}{2} x (x - 1) \ell(u,v)^2 
        + \caO(\epsilon^{2})
        \right) 
        \int Dz \, 
        \left(1 
        + \sqrt{\epsilon} z \beta x \ell(u,v)
        + 2 \sqrt{\epsilon} z \beta ( 1- \ell(u, v))
        \right.\\&\quad\left.
        + 2 \epsilon z^2 \beta^2 x \ell(u,v) ( 1- \ell(u, v))
        + \epsilon \frac{\beta^2}{2} x \ell(u,v) \left( z^2 - 1 + (x - 1) z^2 \ell(u,v) \right)
        - \epsilon \beta^2 (1-\ell(u,v)) \left[1 - 2 z^2 + 3 z^2 \ell(u,v) \right]
        \right.\\&\quad\left.
        + \caO(\epsilon^{3/2})
        \right) 
        \\&= 
        \ell(u, v)^2
        \left(1 
        - \epsilon \frac{\beta^2}{2} x (x - 1) \ell(u,v)^2 
        + \caO(\epsilon^{2})
        \right) 
        \times 
        \\&\quad\times 
        \left(1 
        + 2 \epsilon \beta^2 x \ell(u,v) ( 1- \ell(u, v))
        + \epsilon \frac{\beta^2}{2} x (x - 1) \ell(u,v)^2
        - \epsilon \beta^2 (1-\ell(u,v)) \left( 3 \ell(u,v) -1 \right)
        + \caO(\epsilon^{2})
        \right) 
        \\&= 
        \ell(u, v)^2
        \left(1 
                + 2 \epsilon \beta^2 x \ell(u,v) ( 1- \ell(u, v))
                - \epsilon \beta^2 (1-\ell(u,v)) \left( 3 \ell(u,v) -1 \right)
        + \caO(\epsilon^{2})
        \right) 
    \end{split}
\end{equation}

Now 
\begin{equation}
    \begin{split}
        N^{\rm 2RSB}_z&(u, v)^{p^{\rm 2RSB}_1 / p^{\rm 2RSB}_2} =
        N^{\rm 2RSB}_z(u, v)^{p_1 / x} \\
        \\&=
        \left(1+e^{\beta H_{\rm 1-RSB}(u, v)}\right)^{p_1}
        \left(1 
        + \epsilon \frac{\beta^2}{2} x (x - 1) \ell(u,v)^2 
        + \caO(\epsilon^{2})
        \right)^{p_1 / x}
        \\&=
        \left(1+e^{\beta H_{\rm 1-RSB}(u, v)}\right)^{p_1}
        \left(1 
        + \epsilon \frac{\beta^2}{2} p_1 (x - 1) \ell(u,v)^2 
        + \caO(\epsilon^{2})
        \right)
         \, .
    \end{split}
\end{equation}
so that 
\begin{equation}
    \begin{split}
        N^{\rm 2RSB}_v(u) &= \int Dv \, N^{\rm 2RSB}_z(u, v)^{p^{\rm 2RSB}_1 / p^{\rm 2RSB}_2} 
        \\&= 
        \int Dv \, 
        \left(1+e^{\beta H_{\rm 1-RSB}(u, v)}\right)^{p_1}
        \left(1 
        + \epsilon \frac{\beta^2}{2} p_1 (x - 1) \ell(u,v)^2 
        + \caO(\epsilon^{2})
        \right)
        \\&= N_v(u) \left(1 + \epsilon \frac{\beta^2}{2} p_1 (x - 1)  \angavg{\ell(u,v)^2 }_v + \caO(\epsilon^{2})\right) 
    \end{split}
\end{equation}
and 
\begin{equation}
    \begin{split}
        &\angavg{\angavg{\ell(u,v,z)^2}^{\rm 2RSB}_z}^{\rm 2RSB}_v = 
        N_v(u)^{-1} \left(1 - \epsilon \frac{\beta^2}{2} p_1 (x - 1)  \angavg{\ell(u,v)^2 }_v + \caO(\epsilon^{2})\right) 
        \times \\&\quad\times \int Dv \, 
        \left(1+e^{\beta H_{\rm 1-RSB}(u, v)}\right)^{p_1}
        \left(1 
        + \epsilon \frac{\beta^2}{2} p_1 (x - 1) \ell(u,v)^2 
        + \caO(\epsilon^{2})
        \right)
        \times\\&\quad\times
        \ell(u, v)^2
        \left(1 
        + 2 \epsilon \beta^2 x \ell(u,v) ( 1- \ell(u, v))
        - \epsilon \beta^2 (1-\ell(u,v)) \left( 3 \ell(u,v) -1 \right)
        + \caO(\epsilon^{2})
        \right) 
        \\&= 
        \left(1 - \epsilon \frac{\beta^2}{2} p_1 (x - 1)  \angavg{\ell(u,v)^2 }_v + \caO(\epsilon^{2})\right) 
        \times \\&\quad\times \angavg{
        \ell(u, v)^2
        \left(1 
        + \epsilon \frac{\beta^2}{2} p_1 (x - 1) \ell(u,v)^2 
        + 2 \epsilon \beta^2 x \ell(u,v) ( 1- \ell(u, v))
        - \epsilon \beta^2 (1-\ell(u,v)) \left( 3 \ell(u,v) -1 \right)
        + \caO(\epsilon^{2})
        \right) }_v
        \\&= 
        \left(1 - \epsilon \frac{\beta^2}{2} p_1 (x - 1)  \angavg{\ell(u,v)^2 }_v + \caO(\epsilon^{2})\right) 
        \times \\&\quad\times \angavg{
        \ell(u, v)^2
        + \epsilon \frac{\beta^2}{2} p_1 (x - 1) \ell(u,v)^4
        + 2 \epsilon \beta^2 x \ell(u,v)^3 ( 1- \ell(u, v))
        - \epsilon \beta^2 (1-\ell(u,v)) \left( 3 \ell(u,v) -1 \right) \ell(u,v)^2
        + \caO(\epsilon^{2})
         }_v
         \\&= 
         \left(1 - \epsilon \frac{\beta^2}{2} p_1 (x - 1)  \angavg{\ell(u,v)^2 }_v + \caO(\epsilon^{2})\right) 
         \times \\&\quad\times 
         \left(
         \angavg{\ell(u, v)^2   }_v
         + \epsilon \frac{\beta^2}{2} p_1 (x - 1) \angavg{\ell(u,v)^4   }_v
         + 2 \epsilon \beta^2 x \angavg{\ell(u,v)^3    }_v
         - 2 \epsilon \beta^2 x \angavg{\ell(u,v)^4    }_v
         - 3 \epsilon \beta^2 \angavg{ \ell(u,v)^3   }_v
         \right.\\&\qquad\quad\left.
         + \epsilon \beta^2 \angavg{\ell(u,v)^2   }_v
         + 3 \epsilon \beta^2 \angavg{\ell(u,v)^4    }_v
         - \epsilon \beta^2 \angavg{\ell(u,v)^3   }_v
         + \caO(\epsilon^{2})
         \right)
         \\&= 
         \left(1 - \epsilon \frac{\beta^2}{2} p_1 (x - 1)  \angavg{\ell(u,v)^2 }_v + \caO(\epsilon^{2})\right) 
         \times \\&\quad\times 
         \left(
         \angavg{\ell(u, v)^2   }_v
         + \epsilon \frac{\beta^2}{2} p_1 (x - 1) \angavg{\ell(u,v)^4   }_v
         + 2 \epsilon \beta^2 x \angavg{\ell(u,v)^3    }_v
         - 2 \epsilon \beta^2 x \angavg{\ell(u,v)^4    }_v
         - 3 \epsilon \beta^2 \angavg{ \ell(u,v)^3   }_v
         \right.\\&\qquad\quad\left.
         + \epsilon \beta^2 \angavg{\ell(u,v)^2   }_v
         + 3 \epsilon \beta^2 \angavg{\ell(u,v)^4    }_v
         - \epsilon \beta^2 \angavg{\ell(u,v)^3   }_v
         + \caO(\epsilon^{2})
         \right)
         \\&= 
            \angavg{\ell(u,v)^2}_v
            + \epsilon \frac{\beta^2}{2} \left[
            (6 + p_1 x - p_1 - 4 x)  \angavg{\ell(u,v)^4}_v
            + 4 (x - 2) \angavg{\ell(u,v)^3}_v
            + 2 \angavg{\ell(u,v)^2}_v
            + p_1 (1-x) \angavg{\ell(u,v)^2}_v^2
            \right]
         + \caO(\epsilon^{2})
    \end{split}
\end{equation}
so that the SP equation for $q_2$ reads
\begin{equation}
    \begin{split}
        q_2 + \epsilon &= \angavg{\angavg{\angavg{\ell(u,v,z)}^{\rm 2RSB}_z}^{\rm 2RSB}_v}^{\rm 2RSB}_u
        \\&= 
        \angavg{\angavg{\ell(u,v)^2}_v}_u
        + \epsilon \frac{\beta^2}{2} \left[
        (6 + p_1 x - p_1 - 4 x)  \angavg{\angavg{\ell(u,v)^4}_v}_u
        + 4 (x - 2) \angavg{\angavg{\ell(u,v)^3}_v}_u
        \right.\\&\qquad\left.
        + 2 \angavg{\angavg{\ell(u,v)^2}_v}_u 
        + p_1 (1-x) \angavg{\angavg{\ell(u,v)^2}_v^2}_u
        \right]
     + \caO(\epsilon^{2})
    \end{split}
\end{equation}

The threshold for the linear stability is that the perturbation decreases as the fixed point iterations go on, which happens if 
\begin{equation}
    \begin{split}
        \frac{\beta^2}{2} \left[
        (6 + p_1 x - p_1 - 4 x)  \angavg{\angavg{\ell(u,v)^4}_v}_u
        + 4 (x - 2) \angavg{\angavg{\ell(u,v)^3}_v}_u
        + 2 \angavg{\angavg{\ell(u,v)^2}_v}_u
        + p_1 (1-x) \angavg{\angavg{\ell(u,v)^2}_v^2}_u
        \right] < 1
    \end{split}
\end{equation}
Notice that this condition depends on $x$.
Isolating $x$ we obtain 
\begin{equation}
    \begin{split}
        (6 - p_1)  \angavg{\angavg{\ell(u,v)^4}_v}_u
        - 8 \angavg{\angavg{\ell(u,v)^3}_v}_u
        + 2 \angavg{\angavg{\ell(u,v)^2}_v}_u
        + p_1 \angavg{\angavg{\ell(u,v)^2}_v^2}_u
        \\
        + x \left[ 
            p_1 \left( \angavg{\angavg{\ell(u,v)^4}_v}_u - \angavg{\angavg{\ell(u,v)^2}_v^2}_u \right)
            + 4 \left( \angavg{\angavg{\ell(u,v)^3}_v}_u - \angavg{\angavg{\ell(u,v)^4}_v}_u \right)
        \right]
         < \frac{2}{\beta^2} \, .
    \end{split}
\end{equation}
The most stringent stability condition is given by $x=1$, giving
\begin{equation}
    \begin{split}
        \beta^2 \left[
        \angavg{\angavg{\ell(u,v)^4}_v}_u
        - 2 \angavg{\angavg{\ell(u,v)^3}_v}_u
        + \angavg{\angavg{\ell(u,v)^2}_v}_u
        \right] < 1 \, .
    \end{split}
\end{equation}

\section{The non-symmetric models}\label{sec.SM1-RSBnonSymmetric}

    In this section we discus in more detail the relationship between symmetric and non-symmetric versions of the MAS problem. 
    We start by defining three versions of the MAS problem, based on the properties of the matrix $J$ and of the set of submatrices considered as microstates (the notation follows the definitions of the Main Text):
    \begin{itemize}
    \item \textbf{Rectangular MAS:} $\Nr$, $\Nc$, $\kr$ and $\kc$ are unconstrained. This is the problem relevant for applications \cite{LAS, biclustering_survey}. We define $\aout = \Nr / \Nc$ and $\ain= \kr / \kc$.
    \item \textbf{Square MAS of a square random matrix:} $\Nr = \Nc = N$ and $\kr = \kc = k$, the matrix $J$ is random. This version is studied in the mathematical literature \cite{Bhamidi2017, Gamarnik2018, Cheairi}.
    \item \textbf{Principal MAS of a symmetric random matrix:}   $\Nr = \Nc = N$, $\kr = \kc = k$, $J = J^T$ is symmetric and we consider only \textit{principal submatrices}, i.e. submatrices for which $I_r = I_c$.
\end{itemize}
    In this Section, we provide a mapping of the Rectangular MAS problem onto a bipartite SK model, and compute the associated free entropy using replica theory. 
    We show that in the square case $\alpha_{\rm out} = \alpha_{\rm in} = 1$, the variational free entropy admits a symmetric saddle-point.
    At the symmetric saddle-point, the state equations and the thermodynamic observables of the "Square MAS of a square random matrix" coincide with those of the "Principal MAS of a symmetric random matrix", suggesting that the thermodynamical properties of the two models coincide in the thermodynamic limit.
    This justifies the comparison between the phase diagram we obtain for the "Principal MAS of a symmetric random matrix" and the results in the literature for the "Square MAS of a square random matrix" \cite{Bhamidi2017, Gamarnik2018, Cheairi}.
    We remark that the symmetric saddle-point does not exist for $\alpha_{\rm out}, \alpha_{\rm in} \neq 1$, suggesting that the Rectangular MAS may have a richer phase diagram than the Square/Principal MAS.
    We also observe that bipartite SK models are much harder to study rigorously, and as far as we know no proof supporting replica conjectures is available in this case.

\subsection{The rectangular model}

We can study the non symmetric model by considering the energy function
\begin{equation}
    E(\sigma, \tau) 
    = \frac{1}{\Nr^{1/4}\Nc^{1/4}} \sum_{i=1}^{\Nr}\sum_{j=1}^{\Nc} J_{ij} \sigma_i \tau_j
    = \frac{\aout^{1/4}}{\sqrt{\Nr}} \sum_{i=1}^{\Nr}\sum_{j=1}^{\Nc} J_{ij} \sigma_i \tau_j
\end{equation}
where $\sigma \in \{0, 1\}^\Nr$ and $\tau \in \{0, 1\}^\Nc$.
As in the symmetric case, the Boolean vectors $\sigma$ and $\tau$ encode respectively the row-set and column-set of a submatrix of $J$.
We define the Gibbs measure
\begin{equation}
    p(\sigma, \tau) = \frac{1}{Z_J} \exp\left( \beta E(\sigma, \tau) + \beta \hr \sum_{i=1}^{\Nr} \sigma_i + \beta \hc \sum_{j=1}^{\Nc} \tau_i\right)
\end{equation}
where both magnetic fields are chosen such that 
$\EE \angavg{\sum_{i=1}^{\Nr} \sigma_i} = \kr$ 
and
$\EE \angavg{\sum_{j=1}^{\Nc} \tau_j} = \kc$, 
angular brackets denote averaging over the Gibbs measure at fixed $J$ and $\EE$ denotes averaging over $J$.

We define the aspect ratios $\aout = \Nr / \Nc$ and $\ain = \kr / \kc$ (a priori different).
Without loss of generality, we can consider the case $\aout \geq 1$.
The inner aspect ratio $\ain$ must satisfy the following bounds
\begin{equation}
    \frac{1}{\Nc} \leq \ain \leq \Nr \implies \ain > 0
\end{equation}
where the second inequality holds in the thermodynamic limit $\Nr, \Nc \to \infty$.

We will use the magnetisations defined as
\begin{equation}
    \mr = \frac{\kr}{\Nr}
    \mathand 
    \mc = \frac{\kc}{\Nc} = \frac{\ain \kr}{\aout \Nr} = \frac{\ain}{\aout} \mr
\end{equation}
giving
\begin{equation}
    m_c \in [0, 1] \implies \frac{\ain}{\aout} \mr \leq 1 \implies \mr \leq \frac{\aout}{\ain} \, .
\end{equation}
Thus, the free parameters of the problem will be 
$\beta \in (0,+\infty), \aout \in [1, +\infty), \ain \in (0, +\infty)$ and $\mr \in (0, \aout / \ain)$.

\subsection{Observables}

Define the averaged free entropy as 
\begin{equation}
    \Phi = \lim_{\Nr, \Nc \to \infty} \frac{1}{\sqrt{\Nr \Nc}} \EE_J \log Z_J \, .
\end{equation}
Then, we have the usual grand-canonical decomposition
\begin{equation}
    \Phi = \beta e + s + \beta \hr \mr + \beta \hc \mc
\end{equation}
where $e = E / \sqrt{\Nr \Nc}$ is the average energy density and $s$ the average entropy density.
The energy density can be computed as 
\begin{equation}
    e = \del_\beta \left( \Phi - \beta \hr \mr - \beta \hc \mc \right) \, .
\end{equation}
Finally, the submatrix average is given by
\begin{equation}
    A(\sigma, \tau)
    = \frac{1}{\mr \mc \Nr \Nc} \sum_{i=1}^{\Nr}\sum_{j=1}^{\Nc} J_{ij} \sigma_i \tau_j
    = \frac{1}{\mr \mc \Nr^{3/4}\Nc^{3/4}} E(\sigma, \tau)
    \end{equation}
giving the following relation between the average intensive submatrix average and the average energy density
\begin{equation}
    a = A \Nr^{1/4}\Nc^{1/4} = \frac{1}{\mr \mc} e \, .
\end{equation}

\subsection{Computation of the free entropy using replica theory}

The replicated partition function reads (remember that $\sigma^2 = \sigma$ as $\sigma = 0,1$ and similar for $\tau$)
\begin{equation}
    \begin{split}
        \EE_J Z^n 
        &= 
        \Tr_{\sigma, \tau} 
        \exp\left[ \beta \hr \sum_a \sum_{i=1}^{\Nr} \sigma_i^a + \beta \hc \sum_a \sum_{j=1}^{\Nc} \tau_i^a \right]
        \prod_{i, j = 1}^N 
        \int DJ_{ij} \, 
        \exp
        \left[ \frac{\beta}{\Nr^{1/4}\Nc^{1/4}} \sum_{i=1}^{\Nr}\sum_{j=1}^{\Nc} J_{ij} \sigma_i \tau_j \right] 
        \\&=
        \Tr_{\sigma, \tau} \exp\left[ 
            \beta \hr \sum_a \sum_{i=1}^{\Nr} \sigma_i^a + \beta \hc \sum_a \sum_{j=1}^{\Nc} \tau_i^a
            + \frac{\beta^2 }{2 \sqrt{\Nr \Nc}} \sum_{a, b} 
            \left( \sum_{i=1}^{\Nr} \sigma^a_i \sigma^b_i \right)
            \left( \sum_{j=1}^{\Nc} \tau^a_j \tau^b_j \right)
        \right]
        \\&=
        \Tr_{\sigma, \tau} \exp\left[ 
            \beta \hr \sum_a \sum_{i=1}^{\Nr} \sigma_i^a + \beta \hc \sum_a \sum_{j=1}^{\Nc} \tau_i^a
            + \frac{\beta^2 }{2 \sqrt{\Nr \Nc}} \sum_{a} 
            \left( \sum_{i=1}^{\Nr} \sigma^a_i  \right)
            \left( \sum_{j=1}^{\Nc} \tau^a_j  \right)
        \right.\\&\qquad\qquad\qquad\left.
            + \frac{\beta^2 }{\sqrt{\Nr \Nc}} \sum_{a < b} 
            \left( \sum_{i=1}^{\Nr} \sigma^a_i \sigma^b_i \right)
            \left( \sum_{j=1}^{\Nc} \tau^a_j \tau^b_j \right)
        \right]
    \end{split}
\end{equation}
Now we enforce the order parameters using delta functions and their exponential representation
\begin{equation}
    \begin{split}
        \EE_J Z^n 
        &=
        \int \prod_a d\mr^a \, d\mc^a \, \prod_{a<b} d\qr^{ab} \, d\qc^{ab} \, \times
        \\&\qquad 
        \exp\left[ 
            \Nr \beta \hr \sum_{a} \mr^a 
            + \Nc \beta \hc \sum_{a} \mc^a
            + \sqrt{\Nr \Nc} \frac{\beta^2}{2} \sum_{a} \mr^a \mc^a
            + \sqrt{\Nr \Nc} \beta^2 \sum_{a < b} \qr^{ab} \qc^{ab}
        \right] \times
        \\&\qquad 
        \int \prod_a d\hmr^a \, d\hmc^a \, \prod_{a<b} d\hqr^{ab} \, d\hqc^{ab} \, \times
        \\&\qquad 
        \exp\left[ 
            - \Nr \sum_a \mr^a \hmr^a
            - \Nc \sum_a \mc^a \hmc^a 
            - \Nr \sum_{a<b} \qr^{ab} \hqr^{ab} 
            - \Nc \sum_{a<b} \qc^{ab} \hqc^{ab}
        \right] \times
        \\&\qquad 
        \Tr_{\sigma, \tau} \exp\left[ 
            \sum_{a} \hmr^a \sum_{i=1}^{\Nr} \sigma^a_i 
            + \sum_{a} \hmc^a \sum_{j=1}^{\Nc} \tau^a_j
            + \sum_{a < b} \hqr^{ab} \left( \sum_{i=1}^{\Nr} \sigma^a_i \sigma^b_i \right)
            + \sum_{a < b} \hqc^{ab} \left( \sum_{j=1}^{\Nc} \tau^a_j \tau^b_j \right)
        \right]
        \\&= 
        \int \prod_a d\mr^a \, d\mc^a \, \prod_{a<b} d\qr^{ab} \, d\qc^{ab} \, \times
        \\&\qquad 
        \exp\left[ 
            \Nr \beta \hr \sum_{a} \mr^a 
            + \Nc \beta \hc \sum_{a} \mc^a
            + \sqrt{\Nr \Nc} \frac{\beta^2}{2} \sum_{a} \mr^a \mc^a
            + \sqrt{\Nr \Nc} \beta^2 \sum_{a < b} \qr^{ab} \qc^{ab}
        \right] \times
        \\&\qquad 
        \int \prod_a d\hmr^a \, d\hmc^a \, \prod_{a<b} d\hqr^{ab} \, d\hqc^{ab} \, \times
        \\&\qquad 
        \exp\left[ 
            - \Nr \sum_a \mr^a \hmr^a
            - \Nc \sum_a \mc^a \hmc^a 
            - \Nr \sum_{a<b} \qr^{ab} \hqr^{ab} 
            - \Nc \sum_{a<b} \qc^{ab} \hqc^{ab}
        \right] \times
        \\&\qquad 
        \exp\left[ 
            \Nr \log \Tr_{\sigma} \exp \left( \sum_{a} \hmr^a \sigma^a + \sum_{a < b} \hqr^{ab} \sigma^a \sigma^b \right)
            +
            \Nc \log \Tr_{\tau} \exp \left( \sum_{a} \hmc^a \tau^a + \sum_{a < b} \hqc^{ab} \tau^a \tau^b \right)
        \right] 
    \end{split}
\end{equation}
Thus, we obtain the following variational free entropy (recall $\aout = \Nr / \Nc$)
\begin{equation}
    \begin{split}
        \Phi &= 
        \sqrt{\aout} \beta \hr \sum_{a} \mr^a 
        + \sqrt{\frac{1}{\aout}} \beta \hc \sum_{a} \mc^a
        + \frac{\beta^2}{2} \sum_{a} \mr^a \mc^a
        + \beta^2 \sum_{a < b} \qr^{ab} \qc^{ab}
        \\&\quad
        - \sqrt{\aout} \sum_a \mr^a \hmr^a
        - \sqrt{\frac{1}{\aout}} \sum_a \mc^a \hmc^a 
        - \sqrt{\aout} \sum_{a<b} \qr^{ab} \hqr^{ab} 
        - \sqrt{\frac{1}{\aout}} \sum_{a<b} \qc^{ab} \hqc^{ab}
        \\&\quad
        + \sqrt{\aout} 
        \log \Tr_{\sigma} \exp \left( \sum_{a} \hmr^a \sigma^a + \sum_{a < b} \hqr^{ab} \sigma^a \sigma^b \right)
        + \sqrt{\frac{1}{\aout}} 
        \log \Tr_{\tau} \exp \left( \sum_{a} \hmc^a \tau^a + \sum_{a < b} \hqc^{ab} \tau^a \tau^b \right)
    \end{split}
\end{equation}
to be extremised.
The extremisation condition for the non-hat variable leads to the state equations
\begin{equation}
    \begin{split}
        (\mr^a) \qquad 0 &= 
        \sqrt{\aout} \beta \hr 
        + \frac{\beta^2}{2} \mc^a
        - \sqrt{\aout} \hmr^a
        \implies 
        \mc^a
        =
        \sqrt{\aout} \frac{2}{\beta^2} \left(\hmr^a - \beta \hr\right)
        \\
        (\mc^a) \qquad 0 &=
        \sqrt{\frac{1}{\aout}} \beta \hc 
        + \frac{\beta^2}{2}  \mr^a 
        - \sqrt{\frac{1}{\aout}} \hmc^a 
        \implies
        \mr^a
        =
        \sqrt{\frac{1}{\aout}} \frac{2}{\beta^2} \left(\hmc^a - \beta \hc\right)
        \\
        (\qr^{ab}) \qquad 0 &=
        \beta^2 \qc^{ab}
        - \sqrt{\aout} \hqr^{ab} 
        \implies 
        \qc^{ab}
        = \frac{1}{\beta^2}  \sqrt{\aout} \hqr^{ab} 
        \\
        (\qc^{ab}) \qquad 0 &=
        \beta^2  \qr^{ab} 
        - \sqrt{\frac{1}{\aout}} \hqc^{ab}
        \implies 
        \qr^{ab} 
        = \frac{1}{\beta^2} \sqrt{\frac{1}{\aout}} \hqc^{ab}
    \end{split}
\end{equation}
while that for the hat variables gives 
\begin{equation}
    \begin{split}
        (\hmr^a) \qquad 0 &= 
        - \sqrt{\aout} \mr^a 
        + \sqrt{\aout} \angavg{ \sigma^a }_{\rm r}
        \implies 
        \mr^a = \angavg{ \sigma^a }_{\rm r}
        \\
        (\hmc^a) \qquad 0 &= 
        - \sqrt{\frac{1}{\aout}} \mr^a 
        + \sqrt{\frac{1}{\aout}} \angavg{ \tau^a }_{\rm c}
        \implies 
        \mc^a = \angavg{ \tau^a }_{\rm c}
        \\
        (\hqr^{ab}) \qquad 0 &= 
        - \sqrt{\aout} \qr^{ab} 
        + \sqrt{\aout} \angavg{ \sigma^a \sigma^b }_{\rm r}
        \implies 
        \qr^{ab} = \angavg{ \sigma^a \sigma^b }_{\rm r}
        \\
        (\hqc^{ab}) \qquad 0 &= 
        - \sqrt{\frac{1}{\aout}} \qc^a 
        + \sqrt{\frac{1}{\aout}} \angavg{ \tau^a \tau^b }_{\rm c}
        \implies 
        \qc^{ab} = \angavg{ \tau^a \tau^b  }_{\rm c}
    \end{split}
\end{equation}
where
\begin{equation}
    \angavg{ \sigma^a }_{\rm r} = 
    \frac
    {\Tr_{\sigma} \sigma^a \exp \left( \sum_{a} \hmr^a \sigma^a + \sum_{a < b} \hqr^{ab} \sigma^a \sigma^b \right)}
    {\Tr_{\sigma} \exp \left( \sum_{a} \hmr^a \sigma^a + \sum_{a < b} \hqr^{ab} \sigma^a \sigma^b \right)}
\end{equation}
and similarly for the other averages.
The system of state equations must be solved (in the analytic continuation limit of zero replicas $n\to 0$ as usual for replica computations) for the variables $(\hr, \hc, \hmr, \hmc, \qr, \qc, \hqr, \hqc)$ at fixed $\mr$ and $\mc = \ain / \aout \mr$.
Notice that the magnetisation constraint implies replica symmetry at the level of the magnetisation.

\subsection{The Square MAS of a square random matrix case}

For the "Square MAS of a square random matrix" problem $\ain = \aout = 1$ implying $\mr = \mc$.
It is immediate to see that one solution of the state equations is given by
$\hr = \hc$, $\hmr = \hmc$, $\qr = \qc$ and $\hqr = \hqc$, giving the reduced system of equations
(we are using the implied replica symmetry at the level of the magnetisation here)
\begin{equation}
    \begin{split}
        m
        &=
        \frac{2}{\beta^2} \left(\hm - \beta h\right)
        \\
        q^{ab}
        &= \frac{1}{\beta^2} \hq^{ab} 
        \\
        m &= \angavg{ \sigma^a }
        \\
        q^{ab} &= \angavg{ \sigma^a \sigma^b }
    \end{split}
\end{equation}
which can be further reduced to a set of two equations to be solved for $h$ and $q^{ab}$
\begin{equation}
    \begin{split}
        m &= \angavg{ \sigma^a }
        \\
        q^{ab} &= \angavg{ \sigma^a \sigma^b }
    \end{split}
\end{equation}
where
\begin{equation}
    \angavg{ f(\sigma) } = 
    \frac
    {\Tr_{\sigma} f(\sigma) \exp \left( \left(\beta h + \frac{\beta^2}{2} m\right) \sum_{a} \sigma^a + \sum_{a < b} \beta^2 q^{ab} \sigma^a \sigma^b \right)}
    {\Tr_{\sigma} \exp \left( \left(\beta h + \frac{\beta^2}{2} m\right) \sum_{a} \sigma^a + \sum_{a < b} \beta^2 q^{ab} \sigma^a \sigma^b \right)} \, .
\end{equation}

Notice that the state equations obtained are exactly the same as those obtained for the "Principal MAS of a symmetric random matrix" problem.

The variational free entropy reads ($n$ is the number of replicas)
\begin{equation}
    \begin{split}
        \Phi &= 
        - \frac{\beta^2}{2} n m^2
        - \beta^2 \sum_{a < b} (q^{ab})^2 
        + 2 
        \log \Tr_{\sigma} \exp \left( \left(\beta h + \frac{\beta^2}{2} m\right) \sum_{a} \sigma^a + \sum_{a < b} \beta^2 q^{ab} \sigma^a \sigma^b \right) \, ,
    \end{split}
\end{equation}
the energy density per-replica reads
\begin{equation}
    \begin{split}
        e &=
        \del_\beta( \frac{1}{n} \Phi - 2 \beta h m ) 
        \\&=
        - 2 h m
        - \beta m^2
        - 2 \beta \frac{1}{n} \sum_{a < b} (q^{ab})^2 
        + \frac{2}{n} 
        \angavg{ (h + \beta m) \sum_a \sigma^a + 2 \sum_{a < b} \beta q^{ab} \sigma^a \sigma^b}
        \\&=
        - 2 h m
        - \beta m^2
        - 2 \beta \frac{1}{n} \sum_{a < b} (q^{ab})^2 
        + 2 (h + \beta m) m
        + 4 \beta \frac{1}{n} \sum_{a < b} (q^{ab})^2
        \\&=
        \beta m^2
        + 2 \beta \frac{1}{n} \sum_{a < b} (q^{ab})^2
    \end{split}
\end{equation}
and the average intensive submatrix average (per-replica) is
\begin{equation}
    \begin{split}
        a 
        = \frac{1}{m^2} e 
        = \beta + 2 \frac{\beta}{m^2} \frac{1}{n} \sum_{a < b} (q^{ab})^2 \, .
    \end{split}
\end{equation}

To compare with the "Principal MAS of a symmetric MAS" model, consider the RS ansatz for the overlap order parameter $q^{ab} = q$, giving at leading order for small number of replicas $n$
\begin{equation}
    \begin{split}
        a 
        = \frac{1}{m^2} e 
        = \frac{\beta}{m^2} \left(m^2 - q^2 \right) \, ,
    \end{split}
\end{equation}
which coincides with the RS submatrix average \eqref{eq.energy}.
Similarly one could check that the free entropy and the average submatrix average of the two problems (Square MAS and Principal MAS) coincides at all level of replica symmetry breaking.

\end{document}